\documentclass[12pt, a4paper]{article}
\usepackage[english]{babel}
\usepackage[style=numeric]{biblatex}
\usepackage{amsmath, amssymb, csquotes, amsthm}
\usepackage[colorlinks=true, allcolors=blue]{hyperref}
\usepackage[T1]{fontenc}

\usepackage{csquotes}
\usepackage[capitalise]{cleveref}

\usepackage{thmtools, thm-restate}

\declaretheorem[name=Definition, parent=section]{definition}

\declaretheorem[name=Claim,parent=section]{claim}

\newcommand{\GG}{{\mathcal{G}}}
\newcommand{\FF}{{\mathcal{F}}}
\newcommand{\VV}{{\mathcal{V}}}
\newcommand{\MM}{{\mathcal{M}}}
\newcommand{\NN}{{\mathcal{N}}}
\newcommand{\PP}{{\mathcal{P}}}
\newcommand{\WW}{{\mathcal{W}}}

\newcommand{\Plex}{\PP_{\text{Lex}}}
\newcommand{\Pnmon}{\PP_{\neq}^{\text{Mon}}}

\newcommand{\FFn}{\FF_{\neq}}

\newcommand{\Vmon}{\VV^{\text{Mon}}}

\newcommand{\preq}{\preccurlyeq}

\newcommand{\preqa}{\preq'}
\newcommand{\preqb}{\hat{\preq}}

\newcommand{\preca}{\prec'}
\newcommand{\precb}{\hat{\prec}}

\newcommand{\vecp}[1]{\vec{#1}}
\newcommand{\vecv}{\vecp{v}}
\newcommand{\vecu}{\vecp{u}}

\newcommand{\vecf}[1]{\left(#1_1,\ldots, #1_n\right)}
\newcommand{\vecoper}[2]{\vecf{#1^{#2}}}
\newcommand{\vecord}[1]{\left(\preq_{#1_1},\ldots , \preq_{#1_n}\right)}
\newcommand{\vecordp}[1]{\left(\preq_{\vecp{#1_1}},\ldots , \preq_{\vecp{#1_n}}\right)}

\newcommand{\veco}{\vecf{\preq}}
\newcommand{\vecoa}{\vecf{\preqa}}
\newcommand{\vecob}{\vecf{\preqb}}

\newcommand{\vecopi}{\vecoper{\preq}{\pi}}

\newcommand{\dem}[3]{\text{d}^{#1}_{#2}\left(#3\right)}

\newcommand{\per}[1]{\text{Per}(#1)}
\newcommand{\perM}{\per{\MM}}
\newcommand{\ALLOCs}[1]{{\mathcal{A}_n(#1)}}
\newcommand{\ALLOCsM}{\ALLOCs{\MM}}

\newcommand{\SQ}[2]{g^{#1, #2}}
\newcommand{\SQqp}{\SQ{\vecp{q}}{\vecp{p}}}
\newcommand{\SQqpa}{\SQ{\vecp{q'}}{\vecp{p'}}}

\newcommand{\qp}{(\vecp{q},\vecp{p})}
\newcommand{\qpa}{(\vecp{q'},\vecp{p'})}

\newcommand{\Q}[2]{\text{Q}_{#1,#2}}
\newcommand{\Qnm}{\Q{n}{m}}

\newcommand{\rest}[2]{\left.#1\right\rvert_{#2}}
\newcommand{\restW}[1]{\rest{#1}{\WW}}

\DeclareMathOperator*{\argmax}{arg\,max}

\title{On Truthful Mechanisms without Pareto-efficiency: Characterizations and Fairness}

\author{
    Moshe Babaioff
        \thanks{Hebrew University of Jerusalem. {Emails: \texttt{\{moshe.babaioff, noam.manakermorag\}@mail.huji.ac.il}}}
    \and Noam Manaker Morag
        \footnotemark[\value{footnote}]
}

\date{}

\begin{document}

\maketitle 

\begin{abstract}
We consider the problem of allocating heterogeneous and indivisible goods among strategic agents, with preferences over subsets of goods, when there is no medium of exchange.  This model captures the well studied problem of fair allocation of indivisible goods. Serial-quota mechanisms are allocation mechanisms where there is a predefined order over agents, and each agent in her turn picks a predefined number of goods from the remaining goods. These mechanisms are clearly strategy-proof, non-bossy, and neutral. Are there other mechanisms with these properties? 

We show that for important classes of strict ordinal preferences (as lexicographic preferences, and as the class of all strict preferences), these are the only mechanisms with these properties. Importantly, unlike previous work, we prove the claim even for mechanisms that are not Pareto-efficient. Moreover, we generalize these results to preferences that are cardinal, including any valuation class that contains additive valuations. We then derive strong negative implications of this result on truthful mechanisms for fair allocation of indivisible goods among agents with additive valuations.
\end{abstract}          

\section{Introduction} \label{sec:intro}

One of the central goals of Algorithm Mechanism Design is understanding the extent to which incentives impact the quality of the outcome of algorithms. A central problem that has attracted significant research effort is the problem of fairly allocating indivisible goods, mainly from the algorithmic point of view (when the input is assumed to be known); see, for example, \cite{moulin2004fair, LMMS04, KurokawaPW18,caragiannis2019envy, PR2020, ChaudhuryKMS21, Amanatidis2023, ChaudhuryGM24}. 
The challenge of incorporating incentives for truthful reporting\footnote{{Our work focuses on \emph{dominant strategies}, so ``truthful'' (or ``strategy-proofness'') is in the sense of truthful reporting being a dominant strategy.}} 
when agents' preferences are private information has been addressed only for a special case (partition mechanisms for two agents \cite{AmanatidisBCM17}). 
Unfortunately, that result does not {generalize} beyond that special case, and no results are known beyond that special case, except when also requiring the mechanism to be Pareto-efficient (a strong requirement). This paper aims to make progress on the important and challenging problem of understanding the limits of fair allocation algorithms that incentivize truthful reporting without assuming Pareto-efficiency.

Specifically, we consider the problem of {deterministically} allocating heterogeneous and indivisible goods among agents, 
with either ordinal or cardinal preferences over subsets of goods, with no medium of exchange (no transfers). 
This captures, for example, the prominent model of fair division of heterogeneous and indivisible goods.

In many of these settings, some allocation algorithm is used to decide on the allocation of goods to agents, given the agents' preferences. 
When preferences are private information, agents may act strategically to optimize their allocation. 
For example, consider allocating goods to agents with additive valuations using the Round Robin protocol, {where} agents cyclically pick a good from the remaining goods. 
When all {agents} have the same preference it is always best to pick items in decreasing order of value.  
Yet picking in that order is not a dominant strategy: when preferences are different, there is no reason to take desirable goods early if no other agent wants them.
As in many settings where agents may act strategically, understanding which allocation algorithms are not susceptible to manipulations (are ``strategy-proof'')  is fundamental. Such understanding will also give insights regarding the implications of incentives on the fairness of algorithms.

The seminal work of Amanatidis et al. \cite{AmanatidisBCM17} has studied the fair partition\footnote{Throughout the paper we use the term ``allocation mechanisms'' for mechanisms that may leave goods unallocated, and use ``partition mechanisms'' for mechanisms that must allocate all goods.} of indivisible goods for two strategic agents with private preferences. First, they presented a characterization of truthful {(deterministic)} partition mechanisms for two agents. Then, from the characterization, they derived strong impossibility results for {truthful} fair partition of indivisible goods for two agents with additive valuations. We stress that, somewhat unintuitively, the impossibility for two agents when all goods must be allocated does not immediately imply an impossibility for more agents, or for allocation mechanisms. For example, it does not imply any impossibility for the case of three agents or even for the case of two agents when not all goods must be allocated. When not all goods must be allocated, there may be some mechanism that ``throws away'' some of the goods to create a fairer allocation than when all goods are assigned. Thus, the requirement to partition all goods among the two agents prevents a simple reduction from more than two agents to the two agents' case. We thus ask:

\vspace{0.05in}
\emph{What are the limits of fairness when allocating goods to strategic agents with private preferences?}
\vspace{0.07in}

To address this, it is natural to first ask:

\vspace{0.05in}
\emph{Which mechanisms for allocating goods to strategic agents are strategy-proof (truthful)?}
\vspace{0.05in}

Indeed, similarly to Amanatidis et al. \cite{AmanatidisBCM17} we approach the first question by initially aiming to address the second. So, {we first} seek a characterization of {deterministic} truthful mechanisms and then address the first question using the fact that every mechanism in the limited set of truthful mechanisms has poor fairness guarantees.

Addressing the characterization question is challenging, as even when one restricts attention to consider only the case of  partition mechanisms for two agents, the characterized truthful mechanisms are very involved. 
Currently, extending this characterization without any additional assumptions seems out of reach.
We thus restrict our attention to mechanisms that do not care about the naming of the goods. These are called \emph{neutral} allocation mechanisms in the literature (or, in the context of shares, they are called \emph{name-independent} in \cite{BF22}). Central fairness notions do not depend on good names, for example, the maximin share (MMS)\cite{Budish11}, and the notion of \emph{envy-free-up-to-one-good (EF1)}~\cite{Budish11,LMMS04}.\footnote{\label{footnote:ref-to-mms}See \cref{sec:implications-fair} for the definitions of MMS and EF1.} Similarly, commonly used allocation mechanisms, such as the Round Robin algorithm and the ``cut and choose'' procedure, are neutral. There are other neutral (non-truthful) algorithms for fair allocation, including ones with constant MMS approximation.

When restricting to neutral mechanisms, the characterization of \cite{AmanatidisBCM17} boils down to mechanisms in which, for some $k$, the first agent picks $k$ goods, and the other gets the rest. 
Observe that these mechanisms are actually ``serial-quota'' partition mechanisms for two agents. 
In this paper, we ask for a characterization result beyond two agents and when allowing the mechanism to leave goods unallocated, and as serial-quota mechanisms play a central role in our result, we next discuss these mechanisms.

A \emph{serial-quota mechanism} is a mechanism where there is a predefined order over agents and a predefined quota (number of goods) for each agent, where both the order and quotas are independent of preferences. Each agent, in turn, according to the predefined order, picks a set of goods 
no larger than her quota from the remaining goods.

When agents have strict preferences over sets of goods, these simple mechanisms have several good properties. First, each agent has a dominant strategy (so the mechanism is \emph{strategy-proof}): she should pick a set of size equal to her quota that she likes best (her \emph{demand}) among the remaining goods. Second, the mechanism is \emph{non-bossy}: an agent cannot impact the set received by any other agent, without affecting her own allocation.
Finally, the mechanism is \emph{neutral}: if all agents {identically} rename the goods, the allocation remains the same  (the class is ``permutations-closed'').\footnote{See \cref{sec:mechanisms} for formal definitions of these properties. }

{While serial-quota mechanisms have all these nice properties, they}
are very simple and thus do not possess some other important  properties. For example, every such mechanism might be very unfair to some of the agents. Importantly, for a given serial-quota mechanism, the ordering of the agents and the quotas must be fixed ``a-priori'', that is, they are independent of the reports of the agents. Unfortunately, for whatever quotas are chosen, there will be preferences for which the mechanism’s allocation is very unfair. For example, consider partitioning 1000 goods among five agents with additive valuations over goods. If the quotas of the first four agents are larger than 1, in a setting where the fifth agent only values {five of the goods} (say identically), the mechanism might give her a bundle of no value to her (if the others pick all {five} goods she likes). This outcome seems unfair as she should get at least one of these five goods\footnote{{Formally, since her MMS is not zero, she must get strictly positive value, otherwise the MMS approximation will be zero.} {Additionally, the allocation is not EF1.}}. So the only serial-quota partition mechanism that gives her any value is the one in which each of the first four agents gets one good, and the fifth agent gets all the rest. However, this mechanism is very unfair when all agents value all goods identically.

We thus ask: \emph{Are there any other mechanisms that are strategy-proof, non-bossy, and neutral?} Such mechanisms, in particular, might be able to get better fairness properties. 

A related but easier problem was addressed in \cite{Papai2000}, which showed that for monotone and strict preferences, serial-quota mechanisms are the only mechanisms if, in addition to strategy-proof, non-bossy, and neutral, the mechanism must also be Pareto-efficient. A mechanism is  Pareto-efficient if the allocation it picks  cannot be improved (one agent is strictly better off, while no other is harmed).  

Pareto-efficiency is a very strong requirement, which is violated by many mechanisms, including serial-quota mechanisms for most quotas.
Consider for example allocating three goods to two additive agents, when the first agent picks a single good, and the second agent gets the remaining two goods. 
The allocation of this serial-quota mechanism is not necessarily Pareto-efficient, as evident from the following simple example. Assume that the first agent values the goods at (11,10,10), while the second agent values the goods at (10,1,1). The mechanism results with the first agent getting utility of 11, and the second agent getting utility of 2, 
while if they switch their sets both will be better off:  the first agent will get 20, while the second agent will get 10. 

More generally, \cite{Papai2000} observed that a serial-quota mechanism is Pareto-efficient if and only if all goods are always allocated, and the quotas of all agents except the first are either 0 or 1. Thus, imposing Pareto-efficiency on a serial-quota mechanism implies that not only does the first picker pick goods before others, but she also picks more goods than later agents. This seems very unfair, even in non-pathological examples, as when all goods are identical. We thus move to study truthful mechanisms that are not restricted to be Pareto-efficient, starting with characterizing such mechanisms that are neutral and non-bossy.

\subsection{Our Main Result: Ordinal Mechanism Characterization}
We first consider agents with ordinal preferences and study the problem of allocating a set  $\MM$ of $m$ goods among a set $\NN$ of $n$ agents, each with an ordinal preference over subsets of goods. A {(deterministic)} mechanism is a function that maps  preferences of the $n$ agents to an allocation of the set of goods $\MM$. Preferences over $\MM$ come from a class of ordinal preferences $\PP= \PP(\MM)$, and as items are goods, preferences are weakly monotone.
{F}or our characterization results we assume that preferences are strict\footnote{{Note that although the characterization is for strict preferences, it implies negative results for fair allocation even when preferences {may} not be strict.}}{:} for any two distinct sets, one must be strictly preferred over the other.

Our main {characterization} result is for strict preferences, showing that under mild assumptions about the class of preferences (discussed below), a strategy-proof, non-bossy, and neutral mechanism {must} be a serial-quota mechanism. Our result holds even if the mechanism is allowed not to allocate all goods. {Moreover, o}ur characterization result has significant implications for the setting {where} agents have cardinal valuation functions over {all }subsets of goods, and truthfulness is with respect to reporting of such functions. Although a cardinal valuation function encapsulates more information than an ordinal preference, we show that for important classes of strict valuations over goods, the characterization of truthful, non-bossy, and neutral mechanisms is the same as the one for ordinal preferences.

Our main result applies, in particular, to the well-known class of additive valuation functions (without ties), as well as classes that contain them (such as the classes of sub-additive, sub-modular, and XOS valuation functions, {and also} the class of super-additive valuation functions). A central implication of the characterization is that for additive valuations {(even with ties)}, truthful, non-bossy, and neutral mechanisms have poor fairness properties (we discuss in more details in \cref{sec:intro-imp}).

The two assumptions we make about the class of strict preferences $\PP$ are as follows. First, we assume that all lexicographic preferences are in the class. This holds for many natural classes, such as the class of preferences induced by additive valuations (and any superset of it, as the class of sub-modular or sub-additive valuations) and the larger class of all monotone preferences. Second, we assume that the class is \emph{permutations-closed}; that is, for every preference in the class, it holds that all its permutations are also in the class (See \cref{definition:permuted-preference} for the formal definition and a more detailed discussion). This assumption is satisfied by most natural classes of preferences {(e.g., by the class of additive valuations and by the class of sub-additive valuations)} as it only requires that the names of goods are irrelevant.

For a set of goods $\MM$, we denote the family of classes of strict preferences that are permutations-closed and contain the entirety of the lexicographic preferences by $\FFn(\MM)$. Our main result is a characterization of truthful, non-bossy, and neutral allocation mechanisms as serial-quota mechanisms, for every class in this wide family of classes of strict preferences.

\begin{restatable}{theorem}{strict}
\label{theorem:strict}
Let $f$ be a {deterministic} mechanism that allocates a set  $\MM$ of $m$ goods among $n$ agents with preferences from some permutations-closed class of \emph{strict} preferences $\PP$ that contains the lexicographic preferences ($\PP\in\FFn(\MM)$), when not all goods need to be allocated. 

Mechanism $f$ is truthful, non-bossy, and neutral if and only if $f$ is a serial-quota mechanism. 
\end{restatable}

The characterization is for mechanisms that are truthful, non-bossy, and neutral. We observe that each of the three properties is needed for the characterization to hold. A mechanism that is non-bossy and neutral but \emph{not truthful}, could give each agent in order their ``worst'' set of goods according to the quotas. A mechanism that is truthful and neutral but \emph{not non-bossy}, could let the first agent assign each agent her quota in order, but according to the demand determined by the preference of the \emph{first} agent.
Finally, a mechanism that is truthful and non-bossy but \emph{not neutral}, could simply allocate the goods according to some constant allocation, irrespective of the agents reported preferences.

\paragraph{Proof outline.} 
Before discussing implication of our main result, we present the approach we take proving \cref{theorem:strict}. We prove it in several steps.

We first prove the same characterization for 
{partition mechanisms over} lexicographic preferences (which are always \emph{strict}, meaning that there is never a tie between different subsets). For these preferences, we first show that any truthful, non-bossy, and neutral partition mechanism must be Pareto-efficient (\cref{proposition:Pareto-efficient}). This turns out to be sufficient to imply that the mechanism must be a serial-quota mechanism (\cref{proposition:lexi-case}).  

We then move to consider any class $\PP$ in the family of strict classes $\FFn(\MM)$. We first prove a central lemma, called the Control Lemma (\cref{control}), which is a variant of the Control Lemma proven in \cite{AmanatidisBCM17}, which we later use to prove the characterization result for partition mechanisms over $\PP$. The Control Lemma of \cite{AmanatidisBCM17} is based on two key concepts {defined in \cite{AmanatidisBCM17}}. The first is that of ``strong desire'' and {the} second is of ``control''.

We extend the definition of strong desire to match the setting we study. 
In \cite{AmanatidisBCM17}, ``strong desire'' was defined only for agents with additive valuations over the goods. According to their definition, an agent with an additive valuation function \emph{strongly desires} a set of goods $S\subseteq \MM$ if her value for every individual good in $S$ is greater {than} the sum of the value of all the goods in $\MM\setminus S$. 
We adapt this definition to monotone ordinal preferences. 
An agent with a monotone preference \emph{strongly desires} $S\subseteq \MM$, if for any set $A$ she receives, the agent strictly prefers adding any non-empty subset of $S$ to $A$, over adding the entirety of $\MM \setminus S$ to $A$. In essence, this means the agent will always benefit from exchanging all her goods outside of $S$ for even a single good from $S\setminus A$. For example, an agent with lexicographic preferences over the goods according to some ordering $e_1,\ldots,e_m$ strongly desires every ``prefix'' of this ordering, e.g., the sets $\{e_1\}, \{e_1,e_2\},\{e_1,e_2,e_3\}$ are strongly desired.

The second key concept is of ``control''. Here, our definition remains the same as in \cite{AmanatidisBCM17}: an agent controls a set $S$ for mechanism $f$ if whenever she strongly desires $S$, she gets it whole (and possibly some additional goods). The central result which connects these two concepts is the Control Lemma.

\begin{restatable}{lemma}{control} [Control Lemma]
\label{control}
Let $f$ be a truthful, non-bossy, and neutral partition mechanism of a set of goods $\MM$ among $n$ agents with strict preferences from some permutations-closed class of preferences $\PP$ that contains the lexicographic preferences ($\PP\in\FFn(\MM)$).

Fix a subset $S\subseteq \MM$. If there exists a profile $\veco\in \PP^n$ such that all agents strongly desire $S$, and yet $S\subseteq f_i\veco$ for some  $i\in \NN$, then agent $i$ controls $S$ with respect to $f$.
\end{restatable}

The original version of the Control Lemma presented in \cite{AmanatidisBCM17} was proven only for $n=2$ agents, but without requiring that $f$ be neutral. Additionally, since any partition mechanism for $n=2$ agents is trivially non-bossy (neither agent can affect the {other's} outcome without affecting their own), the requirement of non-bossiness was satisfied (implicitly), but not explicitly stated as a requirement. We require it and only consider mechanisms that are non-bossy. Finally, we prove the lemma for a wide range of ordinal preference classes, while the original version was proven for additive valuation functions. {Our proof of the Control Lemma follows a vastly different approach to that of \cite{AmanatidisBCM17}: proving this result directly for $n\geq 3$ agents appears to be very hard, thus {we} {present an alternative approach which} makes use of an inherent connection between ``strong desire'' and lexicographic preferences to prove the lemma, using the characterization for lexicographic preferences shown earlier.} 
Thus, our version of the lemma is weaker for $n=2$ agents, but our result holds for all $n$ and for a wider range of preference classes. We also present counterexamples showing that all three properties of the mechanism (being truthful, non-bossy, and neutral)  are necessary for $n\geq 3$ agents (\cref{lemma:control-tight}).

{We use the Control Lemma to prove} 
the characterization for the case of partition mechanisms when the permutations-closed class of preferences is strict and contains the lexicographic preferences. We then show that the result can be extended to get the same characterization (as serial-quota mechanisms) not only for partition mechanisms but also for allocation mechanisms (that do not have to allocate all goods), using a reduction argument.

\subsection{Implications of Our Main Results}\label{sec:intro-imp}

We next discuss some implications of our characterization result for ordinal preferences, first to derive characterization results for cardinal preferences (valuations), and then also to derive strong impossibilities for truthful fair allocation of indivisible goods to agents with additive valuations (for mechanisms that are non-bossy and neutral), {even with ties}. 
These impossibilities are in sharp contrast to known results, which demonstrate that good fairness properties can be guaranteed by allocation algorithms (which are not truthful).

\subsubsection{Characterization Results for Cardinal Preferences} 

Let us start with the implied characterization results for cardinal preferences. \emph{Cardinal preferences}, otherwise known as \emph{valuation functions} (or simply \emph{valuations}), are more nuanced than ordinal preferences. 
While ordinal preferences only  provide information on the relative ranking of different subsets of goods, valuation functions provide a real number for each subset, capturing the (cardinal) utility of the subset. Such information allows inter-set comparison of the form \emph{``by how much is one subset better than the other''} and makes the notion of \emph{utility approximation} meaningful (we use that later when discussing MMS {approximation}). 
Each cardinal preference induces a unique ordinal preference, but the opposite is not true. Despite this disparity in information, we prove that our characterization for strict ordinal preferences applies to strict cardinal preferences as well. Formally:

\begin{restatable}{application}{cardinalstrict}
\label{application:cardinal-strict}
    Let $f$ be an allocation mechanism of a set of goods $\MM$ among a set $\NN$ of $n$ agents with \emph{valuations} from a class $\VV$ of \emph{strict} (and monotone) valuation functions that is permutations-closed and induces every lexicographic preference, when not all goods need to be allocated.
    
    Mechanism $f$ is truthful, non-bossy, and neutral if and only if $f$ is a serial-quota mechanism on ${\VV}$.
\end{restatable}

This result is significant as  serial-quota mechanisms cannot make use of the cardinal aspect of the {agents' preferences}: 
for each agent, only the {agent's most preferred set} matters (but not the cardinal value of the set - the intensity of the preference). 
In particular, the additional information provided by the cardinal preference cannot  be used to optimize the allocation.

We also present an extension of this result to classes of valuations that might not be strict. 
For a valuation class $\VV$ of weakly-monotone valuation functions, 
we denote by $\VV_{\neq}$ the subset of strict valuations in $\VV$. 
{For example, if $\VV$ is the class of all additive valuations over goods, the class $\VV_{\neq}$ is the class of additive valuations over goods for which two different subsets of goods never have the same utility (no ties).
  }

\begin{restatable}{corollary}{cardinalnonstrict}
\label{corollary:cardinal-nonstrict}
    Let $f$ be a cardinal allocation mechanism of a set of goods $\MM$ among         a set of $n$ agents with \emph{weakly-monotone valuations}  from a class $\VV$ that is permutations-closed and induces every lexicographic preference, when not all goods need to be allocated. Mechanism $f$ is truthful, non-bossy, and neutral on $\VV_{\neq}$ if and only if $f$ is a serial-quota mechanism on $\VV_{\neq}$.
\end{restatable}
This corollary implies strong negative results on fair truthful mechanisms, as we discuss next.

\subsubsection{Fairness of Truthful Mechanisms} 
For a class of valuations that contains valuations with ties, the above result shows that we know how every cardinal allocation mechanism that is truthful, non-bossy, and neutral must behave on the sub-class of \emph{strict} valuations: it must behave like a serial-quota mechanism on these valuations. In particular, this implies that any such mechanism must behave like a serial-quota mechanism on the sub-class of additive valuations that are strict. As it is known \cite{AmanatidisBM16} that every serial-quota mechanism has poor MMS approximation\textsuperscript{\ref{footnote:ref-to-mms}} 
of only $\frac{1}{\lfloor{\frac{m-n+2}{2}}\rfloor}$, we derive the following strong negative result for fair allocation of indivisible goods.

\begin{restatable}{application}{mms}
{
Consider $n\geq 2$ agents with additive valuations over $m\geq n+2$ goods. Then for any $\rho>\frac{1}{\lfloor{\frac{m-n+2}{2}}\rfloor}$, any cardinal allocation mechanism which is truthful, non-bossy, and neutral on the sub-class of strict additive valuations cannot be $\rho$-MMS.
}

\end{restatable}
The simple mechanism for strict additive valuations in which each agent picks one good in her turn and the last agent receives all remaining goods is $\rho$-MMS for $\rho=\frac{1}{\lfloor{\frac{m-n+2}{2}}\rfloor}$, and this mechanism is a serial-quota mechanism so it is  truthful, non-bossy, and neutral.  Thus, this result shows that no mechanism which is truthful, non-bossy, and neutral on additive valuations can have a better MMS approximation than this trivial mechanism.

We thus see that the impact of truthfulness on performance is significant: in contrast to our negative result showing that $\rho$ is at most $O(1/m)$ for truthful mechanisms, allocation algorithms (that are not truthful but are neutral and non-bossy) give $\rho$ that is a constant.\footnote{Getting the full MMS is impossible, by \cite{KurokawaPW18}. A sequence of works explored the best fraction obtainable, with the state of the art being slightly above $3/4$ \cite{akrami2023breaking}.} 
For example, $\rho=1/2$ can be obtained by the ``bag-filling'' algorithm presented in \cite{Garg2019}. This algorithm is a simple greedy allocation algorithm (that is non-truthful) which is non-bossy and neutral.

We next consider the fundamental envy-based fairness notion for indivisible items, that of  \emph{envy-free-up-to-one-good (EF1)}.
We observe that when preferences are additive (or any class that contains all strictly additive valuations), 
a serial-quota mechanism that is EF1 is very limited in form:  
all agents but the last must receive a quota of one good, while the last
receives either one or two goods. 
In particular, all but at most $n+1$ goods must be left unallocated (ruling out partition mechanisms when $m\geq n+2$).
We use our characterization to conclude the following negative result regarding allocation of truthful, non-bossy, and neutral allocation mechanisms that are EF1:

\begin{restatable}{application}{addef}
    {Consider the class of additive valuations.
    If an EF1 cardinal allocation mechanism is truthful, non-bossy, and neutral
      on the sub-class of strict additive valuations,
    then the mechanism         {allocates at most $n+1$ goods on strict additive profiles.}}  
\end{restatable}

We thus see that the impact of truthfulness on the performance of neutral and non-bossy algorithms is {striking}: our result shows that {any} truthful mechanism that is EF1 must throw away almost all goods (throw at least $m-(n+1)$ goods), {while} the Round Robin algorithm can partition \emph{all} goods and is EF1 (on the reported preferences). {Moreover, in the case that all items are {nearly} identical {(no ties)}, every agent can only get at most $2$ items, a negligible fraction of the $\lfloor \frac{m}{n} \rfloor$ items she gets in the best EF1 allocation.}

\subsubsection{Technical Contributions} \label{sec:tech}

This section details our paper’s main methodological contributions.

\paragraph{Lexicographic Preferences as a Stepping Stone}
We establish our main result {(\cref{theorem:strict})} through a two-step strategy. First, we provide a complete characterization for the restricted class of \emph{lexicographic} preferences. Second, we generalize this result to any {permutations-closed} class of strict and monotone preferences (ordinal or cardinal) that contains the lexicographic preferences. The generalization relies on an extended version of the \emph{Control Lemma} (\cref{control}), originally proved for two-agent partition mechanisms \cite{AmanatidisBCM17}. This lemma makes use of the concept of ``strong desire'', wherein an agent desires each good in a given set $S$ individually more than the entirety of the goods outside this set. The Control Lemma is deeply tied to lexicographic preferences through this notion, as an agent with a lexicographic preference \emph{strongly desires} any set that is a prefix of her ordering over the goods. Thus, we prove the extended Control Lemma by converting general preference profiles to lexicographic preference profiles, and then making use of the full characterization shown for lexicographic preferences earlier. This method of utilizing lexicographic preferences as a stepping stone may be applicable in future work.

\paragraph{Characterization without Pareto efficiency}
Earlier results achieve similar characterizations only under Pareto efficiency (see \cref{sec:related}). In the lexicographic domain, however, Pareto efficiency follows from truthfulness, non-bossiness, and neutrality; we can therefore use Pareto Efficiency in the lexicographic domain ``for free''. This implication allows us to adapt arguments from \cite{Hatfield} to prove the main result for lexicographic preferences, without explicitly requiring Pareto Efficiency. Because the subsequent generalization step does not require the assumption of  Pareto efficiency, our final characterization for general monotone preferences—both ordinal and cardinal—dispenses with that assumption entirely. 
The key technique here is making use of a strong property (in our case, the property of Pareto efficiency), without assuming it explicitly in general, by only using it in a restricted setting where it is naturally implied by other assumptions.

\paragraph{Allowing Disposal of Goods}
Our analysis covers \emph{allocation} mechanisms that may leave some goods unassigned. We accomplish this by reducing any \emph{allocation} mechanism for $n$ agents to a \emph{partition} mechanism for $n+1$ agents by introducing a dummy agent. This agent receives all leftover goods, regardless of her reported preference. We show that our assumptions of truthfulness, non-bossiness, and neutrality of the $n$ agents allocation mechanism, carry over to this newly defined partition mechanism. Since the dummy agent's allocation does not depend on her reported preference, she must appear last in the serial-quota partition mechanism, ensuring that the original allocation mechanism retains the serial-quota form. 
We note that this reduction technique is rather general and only requires the preservation of the assumptions (it is not specific to the particular properties of truthfulness, non-bossiness, or neutrality). Thus, this technique is applicable to any setting where the required assumptions continue to hold when another dummy agent is added. 
Crucially, this technique requires a characterization for arbitrary $n$, and therefore is not applicable to the characterization of partition mechanisms for two agents presented in \cite{AmanatidisBCM17}.

\subsection{Future Work}\label{sec:future}

We discuss several directions for future work.

The first direction is of finding characterization results for truthful, non-bossy, and neutral mechanisms for classes of preferences not covered by our results. In particular, we focus on strict preferences, and do not handle preferences with ties.
One may first look for a characterization result for truthful, non-bossy, and neutral mechanisms over the entire class of weakly-monotone preferences. These preferences pose the unique challenge that an agent is not even guaranteed to demand their entire quota: an agent may be offered a quota of several goods, but some sets in the demand will be smaller than the quota.

Another direction for future work is that instead of focusing on characterization results, focus on proving fairness  
impossibilities for truthful mechanisms with fewer assumptions. 
Our characterization assumes that the mechanism is not only truthful, but also non-bossy and neutral, and so our fairness impossibilities only hold for these mechanisms. It would be interesting to get similar negative results without assuming neutrality (even when explicitly assuming non-bossiness), but for more than two agents.
Alternatively, one can try to obtain results for neutral mechanisms that are not non-bossy.

\subsection{Related Work}
\label{sec:related}

The literature relevant to this article originates from two primary fields of study, {one on mechanism design and the other on fair allocation of indivisible goods}.

\subsubsection{Mechanism Design for Indivisible Goods}
Our paper is part of the vast literature on the design of allocation mechanisms in settings with private preferences (without transfers). We next discuss several papers in this literature that are most related to ours.  

 Serial-quota mechanisms were first presented by Papai in \cite{Papai2000}, where it was shown that when {agents' preferences} are strict and quantity-monotone (more is always better), an allocation mechanism {is} strategy-proof, non-bossy, Pareto-efficient, and neutral if and only if it is a serial-quota mechanism. A similar result was shown when {agents' preferences} are strict and monotone: an allocation mechanism {is} strategy-proof, non-bossy, Pareto-efficient, and neutral if and only if it is a \emph{constrained} serial-quota mechanism, where all agents but the first receive a quota of at most one good. The main contribution of our work over \cite{Papai2000} is removing the strong requirement of Pareto-efficiency, {thus permitting the use of any chosen quotas}.

Some papers examine a model where the mechanism must abide some preexisting \emph{quota-system}, which determines exactly how many items each agent must receive. Hatfield \cite{Hatfield} studies the case when each agent must receive exactly $Q$ goods, and agents are assumed to have responsive preferences over items, and showed that an allocation mechanism is strategy-proof, Pareto-efficient, and non-bossy if and only if it is a sequential\footnote{While ``serial'' constrains the order and quotas to be independent of preferences, ``sequential'' allows some dependence\cite{Hatfield}.}-quota mechanism, and that adding the assumption of neutrality further restricts the characterization to serial-quota mechanisms. Hosseini and Larson \cite{Hosseini2016} study the case when there is some quota-system and {agents' preferences} are lexicographic,  and shows that an allocation mechanism is strategy-proof, non-bossy, neutral, and Pareto $C$-efficient\footnote{An allocation mechanism is Pareto $C$-efficient if it's allocations are Pareto optimal among all allocations that allocate exactly $C$ items in total.}  if and only if it is a serial-quota mechanism.

Another prominent model studied in the field of allocation of indivisible goods is the ``house allocation model''\cite{Shapely1974,Ma1994, Svensson1994, Svensson1999, Papai2000b}, in which each agent can receive only one good.
In \cite{Svensson1999} a characterization result is presented for this model, which is very similar to the one we prove for our model:
it is shown that a mechanism for the house allocation model is truthful, non-bossy, and neutral if and only if it is serially-dictatorial (agents choose their house in turn according to some predefined ordering). Our paper studies the more general and challenging problem, when agents have preferences over subsets and it is possible to allocate \emph{multiple} goods to an agent, sometimes known as the ``multiple assignment'' model.

There have {been} many different characterizations shown for the multiple assignment problem, each with different assumptions about the preferences of the agents and the properties the mechanism is assumed to possess. A common element in all these papers is the requirement that the mechanism be Pareto-efficient\footnote{In some papers (such as \cite{Svensson1999}) Pareto-efficiency is not assumed outright but instead proven from the other assumptions. This is still fundamentally different from our characterization, which includes mechanisms that are not Pareto-efficient.}.

Our characterization consider mechanisms that are ``non-bossy'', a property that is very helpful in proving our results.  
Thomson \cite{Thomson2016} acknowledges its technical appeal, but presents various critiques of this property in different models.
Characterizing truthful mechanisms without it, and understanding the limits of fair truthful mechanisms that are not required to satisfy this property, are two interesting research directions.

\subsubsection{Fair Allocation of Indivisible Goods}

{We} move to consider papers that do not assume {Pareto-efficiency}, mainly in the context of fair allocation.
There is extensive research on fair allocation, much more than we can survey in this section. The reader is referred to the books \cite{brams1996fair,moulin2004fair} and surveys \cite{AzizSurvey2022, fairSurvey2022} for background. We next discuss some {papers} that are most related to ours. 

The work of Amanatidis et al. \cite{AmanatidisBCM17} characterizes truthful partition mechanisms for two agents. As all goods must be allocated and as there are only two agents, the mechanism is trivially non-bossy. Unlike our paper, {the} paper {considers} mechanisms that are not constrained to be neutral, and thus the set of truthful mechanisms is much richer than only serial-quota mechanisms. On the other hand, we are able to present characterizations for more than two agents, and for mechanisms that are allowed not to allocate all goods, by focusing on neutral mechanisms (with an explicit assumption of non-bossiness).
{The earlier work of \cite{AmanatidisBM16} provides some general bounds on the fairness of truthful mechanisms in different models, especially the case when there are two agents.}

In the context of assignment of chores (items of negative value), \cite{ijcai2019p9} has studied the power of truthful {mechanisms}, presenting bounds on MMS approximation of serial-quota mechanisms for both the ordinal and the cardinal models. 
The paper {leaves} gaps between its upper and lower bounds on the obtainable MMS approximation, and it might be possible to close those gaps (likely proving the upper bounds are tight) by presenting a characterization of truthful (non-bossy and neutral) mechanisms for the assignment of chores. 

The issue of dependence on item names was discussed in several papers on fair allocation \cite{BF22,BabaioffF2024}. The paper \cite{BF22} defines a general notion of a ``share'' for agents with equal entitlements, and  requires shares to satisfy  ``name-independence'', which is essentially neutrality, but for shares rather than mechanisms. The paper \cite{BabaioffF2024} extends the definition to agents with arbitrary entitlements, and studies both the case {where shares} are required to be ``name-independent'', and the case they do not.

While \cite{AmanatidisBCM17} studied truthful mechanisms and obtained strong negative results, several other papers have {studied} other notions aimed {at} {incorporating} incentives into fair allocation in different ways. 
Several papers \cite{AmanatidisBFLLR21, AmanatidisBL0R23} have studied Nash equilibrium of the Round Robin procedure (in which {agents} cycle in picking a good each).  Round Robin is a partition mechanism that is non-bossy and neutral but not truthful. 
It was shown \cite{AmanatidisBFLLR21} that there is a Nash equilibrium in the induced game that is EF1. 
Existence of a Nash equilibrium is weaker than dominant strategy truthfulness, and our {results} show that to get a positive result in a strategic setting for EF1 partition by a (non-bossy and neutral) algorithm, one must indeed relax truthfulness.

An alternative approach to addressing incentives was suggested in \cite{BF22}, which introduced the notion of ``self-maximizing'' shares, a property of shares that {incentivizes} ``pessimistic'' agents to report valuation functions truthfully.

\section{The Ordinal Model}\label{sec:model}

We consider the allocation of a set $\MM$ of $m$ indivisible goods to a set $\NN$ of $n$ agents (for $m,n\in \mathbb{N}$). 
Each agent $i\in \NN$ has {an \emph{ordinal preference} (or simply a \emph{preference})} 
{$\preq_i$} over {subsets of}
$\MM$, which defines an order relation on pairs of subsets of $\MM$.
That is, for every pair of subsets $S,T\subseteq \MM$, we write $T\preq S$ (or $S \succcurlyeq T$) when subset $S$ is weakly preferred to $T$ according to the preference $\preq$. 

\subsection{Preferences}
{When considering ordinal preferences, we study preferences that are strict.}

\begin{definition}[{Strict preferences}]
For any pair of subsets $S,T\subseteq \MM$, subset $S$ is \emph{strictly preferred} to $T$ according to the preference $\preq$, if  $T\preq S$ but not $S\preq T$. In this case we {write} $T\prec S$ (or $S\succ T$). A preference $\preq$ is \emph{strict} if  for any two distinct sets, {there is one that is strictly preferred to the other}. Formally, $\preq$ is strict if for every $S,T\subseteq \MM$, if $S\neq T$ then either $S\prec T$ or $T\prec S$. 
\end{definition}

We use $\PP(\MM)$ to denote a \emph{class} (set) of preferences over $\MM$. When $\MM$ is clear from the context, we simplify notation and  denote $\PP(\MM)$ by $\PP$. We use $\veco\in \PP^n$ to denote a \emph{preference profile}, a vector of $n$ preferences in $\PP$, one for each agent.

Throughout this paper we assume that items are \emph{goods}: an agent is never harmed by getting additional items. That is, we assume that all preferences are weakly monotone.  A preference $\preq$ is \emph{weakly monotone} (or simply \emph{monotone}) if {for every $T\subseteq S \subseteq \MM$} it holds that $T\preq S$. 
{Moreover, we will assume that {ordinal} preferences are strict. We denote the set of all \emph{monotone and strict} preferences over $\MM$ by $\Pnmon(\MM)$.}

{We are interested in mechanisms where the numbering of the goods does not affect the outcome (see \cref{definition:neutrality}). To make this formal we next define the notion of \emph{permuted preferences}.} 
Let $\perM$ denote the set of all permutations over $\MM$. For permutation  $\pi\in \perM$ and set $S\subseteq\MM$ we use $\pi(S)$ to denote the set of goods $\left\{\pi(x)\mid x\in S\right\}$. 

\begin{definition}[Permuted preference]
\label{definition:permuted-preference}
    Given a permutation $\pi\in\perM$, for any preference $\preq$, let $\preq^{\pi}$ denote the preference defined as follows. For $S,T \subseteq \MM$:    \[S \preq^{\pi} T \iff \pi^{-1}(S) \preq \pi^{-1}(T) \]
\end{definition}

The above requirement can be equivalently stated: $\pi(S) \preq^{\pi} \pi(T) \iff S\preq T$ for all $S,T\subseteq \MM $. In particular, for a strict preference $\preq$, for each set $S$, the rank of the set $\pi(S)$ according to the preference $\preq^{\pi}$ is the same as the rank of $S$ according to the preference $\preq$.

We say that a set of preferences $\PP$ is \emph{permutations-closed} if it is closed under permutations, that is, for any permutation  $\pi\in \perM$ and $\preq\in \PP$ it holds that the preference $\preq^{\pi}$ is in $\PP$ as well. Throughout this paper we assume that the set of preferences $\PP$ over which an allocation mechanism is defined is always permutations-closed. This assumption is satisfied by many natural classes of preferences, e.g., {the class of} lexicographic preferences, {the class of additive preferences} and {their superclass} of all monotone preferences.

Preferences that are lexicographic play an important {role} in our proofs. A lexicographic preference is defined by an  ordering $(e_1,e_2,\ldots, e_m)$ over goods, and the preference is determined by the ``first good'' in the ordering that belongs to one set but not the other.
\begin{definition}[{Lexicographic preferences}]
\label{def:lexi}
    A preference $\preq$ is \emph{lexicographic} if there exists an ordering $(e_1,e_2, \ldots, e_m)$ of the set $\MM$ of $m$ goods such that for all $S,T\subseteq \MM$, it holds that $T \prec S$ if and only if there exists $j\leq m$ such that $e_j \in S\setminus T$ and for all $k<j$ it holds that $e_k\in S   \iff e_k\in T$. We denote  the lexicographic preference corresponding to the ordering $\vecp{e}=(e_1,\ldots , e_m)$ by $\preq_{\vecp{e}}
    $. We use $\Plex(\MM)$ to denote the class of all lexicographic preferences over $\MM$.
\end{definition}

As any lexicographic preference must be both weakly-monotone and strict, it holds that  $\Plex(\MM)\subsetneq \Pnmon(\MM)$. Moreover, since the permutation of any lexicographic preference must also be lexicographic, the class $\Plex(\MM)$ is permutations-closed {(see \cref{appendix:subsec:lexi} for more discussion).}

Our proof extensively uses orderings of goods {which} are consistent with some order over subsets as follows.
\begin{definition}[{Order consistency}]
    Let $(S_1,S_2, \ldots ,S_k)$ be a partition of the goods $\MM$ into $k\geq 2$ sets. We say that the ordering $\vecp{e}=(e_1,e_2,\ldots,e_m)$ \emph{is consistent with} the order of subsets $(S_1,S_2, \ldots ,S_k)$ if for each pair of indices $i<j$, all goods in $S_i$     {precede those} in $S_j$ in the order $\vecp{e}$. That is, $S_1=\{e_1,e_2,\ldots,e_{|S_1|}\}$  (as sets), $S_2=\{e_{|S_1|+1},\ldots,e_{|S_1|+|S_2|}\}$ and so forth.
\end{definition}

Our paper focuses on  characterization of mechanisms for classes of preferences that belong to some families of classes we define below.
\begin{definition}
Fix a set $\MM$ of goods.
    {We use $\FFn(\MM)$ to  denote the family of all permutations-closed classes of strict preferences that contain the entirety of the lexicographic preferences $\Plex(\MM)$. (As items are goods, every preference in the class must be in $\Pnmon(\MM)$). Formally:
\[\PP\in \FFn(\MM) \iff \Plex(\MM) \subseteq \PP\subseteq \Pnmon(\MM) \text{ and $\PP$ is permutations-closed.}\]

    }
\end{definition}

Our main characterization result gives us complete understanding of the structure of truthful, non-bossy, and neutral mechanisms, not only for the case that the class of preferences is $\Pnmon(\MM)$, but actually for any class of preferences in $\FFn(\MM)$.
Note that  $\Pnmon(\MM)$ includes more preferences than any of its sub-classes. 
{This means that an agent has more potential
``deviations'' when reporting her preferences,} so there are additional truthfulness constraints. Thus, a-priori, the characterization {of truthful mechanisms} for $\Pnmon(\MM)$, might be different than the one for its sub-classes.

\subsection{Mechanisms}\label{sec:mechanisms}

We study (deterministic) allocation mechanisms\footnote{In this paper we focus on \emph{direct revelation} mechanisms. This is without loss of generality, by the revelation principle. We also focus on deterministic mechanisms, and do not consider mechanisms that are allowed to use randomization.} of goods to agents, without monetary transfers, when the preference of each agent is private information to her. Given a profile of (reported) preferences, an allocation mechanism picks an allocation of goods to the agents.  An \emph{allocation} $A=(A_1,\ldots, A_n)$ is a collection of $n$ disjoint subsets of $\MM$. That is, $A_i\subseteq \MM$ for all $i\in \NN$, and $A_i\cap A_j=\emptyset$ for all $i\neq j\in \NN$. For {allocations} it holds that $\cup_{i\in \NN} A_i\subseteq \MM$, allowing goods to be left unallocated. We say that an allocation $A$ is a \emph{partition} of $\MM$ if all goods are allocated, that is, $\cup_{i\in \NN} A_i= \MM$. We use $\ALLOCsM$ to denote the set of all allocations of $\MM$ among $n$ agents.

\begin{definition}[Allocation mechanism]
An \emph{allocation mechanism} $f$ of a set $\MM$ of goods to a set $\NN$ of $n$ agents, each with preferences from $\PP$, is a (deterministic) function that maps a {reported} profile of $n$ preferences
{from} $\PP$, to an allocation $A\in \ALLOCsM$ of the set of goods $\MM$ among the agents. That is, for $f:\PP^n\rightarrow \ALLOCsM$, when the {reported} preference profile is $\veco\in \PP^n$, the set that $f$ allocates to agent $i$ is $A_i=f_i\veco$.
\end{definition}

{
Throughout the paper, when we talk about an allocation mechanism $f$, unless stated otherwise, we assume that  $f:\PP^n\rightarrow \ALLOCsM$ is an allocation mechanism of a set  $\MM$ of $m$ goods  to a set $\NN$ of $n$ agents, each with  preference in the class $\PP$ of strict preferences.
}
Some of the results we prove are for allocation mechanisms that must allocate all the goods, which we call ``partition mechanisms''.
\begin{definition}[Partition mechanism]
    \label{partition-mechanism}
    An allocation mechanism $f$  is a \emph{partition mechanism} if for every preference profile $\veco\in \PP^n$ the allocation $A=f\veco$ is a partition of $\MM$.
\end{definition}

We next formalize some standard properties of mechanisms from the literature. 

\begin{definition}
    [Truthfulness]
    {An allocation} mechanism  $f$ is \emph{truthful} (in dominant strategies), if for every  agent $i\in\NN$, every preference profile $\veco= (\preq_i,\preq_{-i})\in \PP^{n}$, and every alternate {reported}     preference $\preqa_i\in \PP$ of agent $i$, it holds that $ f_i(\preqa_i,\preq_{-i})\preq_i f_i(\preq_i,\preq_{-i})$.
\end{definition}
Intuitively, truthfulness states that  for any profile of preferences, there is no agent that can get a strictly better bundle by misreporting  her preference.

\begin{definition} [Non-Bossiness]
\label{definition:non-bossy}
     {
     Consider a class $\PP$ of strict preferences. {An allocation} mechanism $f$ is \emph{non-bossy} if for every agent $i\in\NN$, every preference profile $\veco= (\preq_i,\preq_{-i})\in \PP^{n}$, and every alternate preference $\preqa_i\in \PP$:
     \[f_i(\preqa_i,\preq_{-i})= f_i(\preq_i,\preq_{-i}) \implies f(\preqa_i,\preq_{-i}) = f(\preq_i,\preq_{-i})\]}
\end{definition}
{
Intuitively, non-bossiness states that in order for {a change in the preference of an} agent to ``affect'' the allocation of other agents, she must affect her own allocation.
}

In this paper we focus on mechanisms that are neutral, these are mechanisms for which the names of the goods are insignificant and permuting the names of the goods makes no difference.

\begin{definition} [Neutrality]
\label{definition:neutrality}
    {
    Consider any permutations-closed class $\PP$ of strict preferences. {An allocation} mechanism $f$ is \emph{neutral} if for every permutation $\pi\in\perM$ and preference profile $\veco\in \PP^n$, it holds that $f\vecopi=\pi(f\veco)$
    }
\end{definition}

{
Intuitively, neutrality states that if all agents agree on a specific ``renaming'' of the goods, and all update their preferences accordingly, 
all agents will get the permutation of their original allocation.}
The goal of this property is to focus on mechanisms that are precluded from being dependent on the arbitrary ``numbering'' of the goods.

An important property of an allocation mechanism is of being Pareto-efficient, only generating Pareto-optimal allocations.  
\begin{definition} [Pareto-efficiency]
    For a given preference profile $\veco$, an allocation $A=(A_1,\ldots, A_n)$ is \emph{Pareto-optimal} if there does not exist an alternate allocation $B=(B_1,\ldots, B_n)$ such that $A_i\preq_i B_i$ for all agents $i\in\NN$ and for at least one agent the inequality is strict. 
    An allocation mechanism $f$ {for agents with preferences from class $\PP$} is \emph{Pareto-efficient}, if for every preference profile $\veco{\in \PP}$ the allocation $A=f\veco$ is Pareto-optimal.
\end{definition}

\subsection{Serial-Quota Mechanisms}
For classes of strict preferences, our characterization relates to serial-quota mechanisms which we define below. We first need to define the demand of an agent with some preference and quota.

\begin{definition}
Given a set of goods $S\subseteq \MM$, a number $k\in[|S|]\cup \{0\}$ and a \emph{strict} preference $\preq\in \PP$, define the \emph{$k${-demand} from $S$ according to $\preq$}, denoted $\dem{\preq}{k}{S}$, as {the (unique) most preferred subset of at most $k$ goods from $S$, according to $\preq$.} Formally\footnote{{We use  $\argmax_{\preq}\GG$ 
to denote the highest ranked set from a family of sets $\GG$ according to \emph{strict} preference $\preq$.}}:
\[\dem{\preq}{k}{S}=\argmax_{\preq}\{T\subseteq S,\space |T|\leq k\} \]
\end{definition}

\begin{definition}
\label{quota-ordering-pair}
    A \emph{quota-vector $\vecp{q}$} for $n$ agents and $m$ goods is a list of non-negative integer  \emph{quotas} $q_1,\ldots, q_n$ such that $\sum_{i=1}^n q_i \leq m$. An \emph{ordering  $\vecp{p}$} of $n$ agents is a list $p_1,\ldots , p_n\in \NN$ of distinct agents. Let $\Qnm$ denote the set of all pairs $\qp$ of quota-vectors and orderings for $n$ agents and $m$ goods, satisfying the following two requirements:
\begin{enumerate}
    \item All {agents with quota zero} must be at the end of the ordering. That is, if $q_i=0$ for some $i\in[n]$, then for all $j>i$, $q_j=0$ as well.
    \item Agents with quota zero are ordered according to their index in the input. That is, letting $i$ be the first index with quota $0$, agents $p_i,\ldots,p_n$ are ordered so that $p_i<\ldots<p_n$.
\end{enumerate}
\end{definition}

The reason for including the two additional requirements is to ensure a unique representation of serial-quota mechanisms that are equivalent. We give more details
{in \cref{appendix:subsec:sq-mech}}.

\begin{definition}
{
    Let $\PP$ be a class of \emph{strict} preferences over a set $\MM$ of $m$ goods. For a pair $\qp\in\Qnm$, we say the mechanism $f$ is a \emph{$\qp$-serial-quota mechanism} for $n$ agents with preferences from class $\PP$, if for any profile of preferences $\veco\in \PP^n$ it holds that for every index\footnote{In the case $i=1$, the first agent simply picks her favorite $q_1$ goods from the entirety of $\MM$.} $i\in [n]$:
        \[f_{p_i}\veco =  \dem{\preq_{p_i}}{q_i}{\MM \setminus \bigcup_{j=1}^{i-1} f_{p_j}\veco }   \]     
        That is, for every $i\in[n]$ in order, agent $p_i$ receives her $q_i$-demand from the goods remaining, according to preference $\preq_{p_i}\in \PP$. We let $\SQqp$ denote the unique serial-quota mechanism for the pair $\qp$.
}
\end{definition}

We say an allocation mechanism $f$ is a \emph{serial-quota mechanism} if there exists a quota-ordering pair $\qp\in\Qnm$ s.t. $f$ is a $\qp$-serial-quota mechanism.
For a quota-ordering pair $\qp\in\Qnm$, we say an allocation mechanism $f$ is a \emph{$\qp$-serial-quota mechanism on $\widehat{\PP}$}, if $f\veco=\SQqp\veco$ for every preference profile $\veco\in \widehat{\PP}^n$.

\section{A Characterization for Strict Ordinal Preferences} \label{sec:strict}
In this section we present a characterization of all truthful, non-bossy, and neutral allocation mechanisms for any permutations-closed class $\PP$ of strict preferences that contains the lexicographic preferences ($\PP\in\FFn(\MM)$), even when not all goods must be allocated. Our characterization shows that such mechanisms must be serial-quota mechanisms:

\strict*

To show that $f$ is a serial-quota mechanism we show that there exist $\qp\in\Qnm$ such that $f\veco=\SQqp\veco$ for every profile $\veco\in\PP^n$. 

Before moving to prove the theorem, let us outline the steps of the proof.
Considering a permutations-closed class of strict preferences over goods, we prove the claim in several steps. 
We first prove that when all preferences are lexicographic, the combination of truthfulness, non-bossiness, and neutrality implies Pareto-efficiency (\cref{proposition:Pareto-efficient}). Then, using Pareto-efficiency,  we prove the characterization for partition mechanisms when preferences are lexicographic (\cref{proposition:lexi-case}). We then move to consider any class of strict preferences that contains the lexicographic preferences. We prove the useful ``Control Lemma'' for these classes (\cref{control}). We then use an induction argument to prove the characterization for partition mechanisms over any such class of strict preferences (\cref{strict-partition}). Finally, we prove \cref{theorem:strict}, extending the result from partition mechanisms to allocation mechanisms using a reduction argument.

In the rest of the section we prove \cref{theorem:strict} according to the above 
outline.

\subsection{Monotonicity}
\label{sec:monotonicity}

The widely used property of \emph{monotonicity} is known to be tightly connected to truthful {and non-bossy} allocation {mechanisms}, and we also use that connection extensively. Intuitively, monotonicity states that if all agents alter their preferences in a way that makes the bundle they got even more preferable, the allocation must remain the same.

\begin{restatable}[Push-up]{definition}{pushupdef}
     We say that preference $\preqa\in \PP$ is a \emph{push-up} of $\preq\in \PP$ for set $S\subseteq \MM$ if any subset $Z\subseteq \MM$ s.t. $Z\preq S$ it holds that $Z\preqa S$. Let $f$ be an allocation mechanism over the preference class $\PP$. We say that the profile $\vecoa\in \PP^n$ is a \emph{push-up} of $\veco\in \PP^n$ for $f$ if for every agent $i\in\NN$, the preference $\preqa_i$ is a push-up of $\preq_i$ for allocation $f_i\veco$.
\end{restatable}

{Intuitively, this states that each agent with preference from $\vecoa$ ``desires'' the set {she is} allocated under $\veco$ even more.}

\begin{restatable}[Monotonicity]{lemma} {monotonicity}
 \label{monotonicity}
    Let $f$ be a truthful and non-bossy allocation mechanism for agents with preferences from a class $\PP \subseteq \Pnmon(\MM)$ of strict preferences. Then for any pair of valuation profiles $\veco,\vecoa\in \PP^n$ such that $\vecoa$ is a push-up of $\veco$ for $f$, it holds that $f\veco=f\vecoa$.
\end{restatable}

This connection between monotonicity and truthful {and non-bossy} allocation {mechanisms} is a well known fact that has been proven before for both ordinal and cardinal preferences \cite{Papai2001, Hatfield}. For the sake of completeness we present a proof using our notations in \cref{appendix:monotonicity}.

\subsection{Partition under Lexicographic Preferences} \label{subsec:lexi}

In this section we prove the main {result} 
for the special (and severely restricted) case when all the preferences are lexicographic, and moreover, when the mechanism is a partition mechanism that must allocate all goods. 

That is, we prove our characterization result (that any truthful, non-bossy, and neutral mechanism $f$ must be a serial-quota mechanism) for a mechanism $f$ that partitions a set $\MM$ of $m$ goods among a set $\NN$ of $n$ agents, each with preference from the class $\Plex(\MM)$ of lexicographic preferences.
{W}e first show that {when preferences are lexicographic,} if a  partition mechanism $f$ is truthful, non-bossy, and neutral, then it must be Pareto-efficient (\cref{sec:lex-PO}). We then use this to prove that it must be a serial-quota mechanism, using arguments similar to those of Hatfield \cite{Hatfield} (\cref{sec:lex-hatfield}).

\subsubsection{Pareto-efficiency}\label{sec:lex-PO}

{The following powerful result allows us to make use of previous work which was done for Pareto efficient mechanisms, without explicitly assuming Pareto efficiency.}

\begin{restatable}{proposition}{paretolexi}
    \label{proposition:Pareto-efficient}     
     Let $f$ be a \emph{partition} mechanism for $n$ agents with preferences from the class $\Plex(\MM)$ of lexicographic preferences. If $f$ is truthful, non-bossy, and neutral, then it is Pareto-efficient.
\end{restatable}

{This result is crucial to our proof, as we later show that the characterization for lexicographic preferences can be generalized to monotone preferences without assuming Pareto efficiency. The proof of this result, along with the necessary definitions and lemmas is given in \cref{appendix:lexi-po}.} With Pareto-efficiency, proving that the mechanism must be a serial-quota mechanism follows from arguments presented in Hatfield \cite{Hatfield}, which we adapt to our problem in the next section.

\subsubsection{Characterization for Partitions under Lexicographic Preferences}
\label{sec:lex-hatfield}

We
{take inspiration} 
from 
\cite{Hatfield} to prove that for lexicographic preferences, if a partition mechanism $f$ is truthful, non-bossy, and neutral, then it must be a serial-quota mechanism. In proving this result we make use of variants of several lemmas that Hatfield proves.  Since the definitions and assumptions of Hatfield are different from ours, the lemmas cannot be used as stated, so we adapt the statements and proofs of the lemmas and the result they imply to our model and definitions, while maintaining the core ideas.

\begin{restatable}{proposition}{lexicase}
\label{proposition:lexi-case}
  Let $f$ be a partition mechanism for $n$ agents with preferences from the class $\Plex(\MM)$ of lexicographic preferences. Mechanism $f$ is truthful, non-bossy, and neutral if and only if $f$ is a $\qp$-serial-quota mechanism.
\end{restatable}

{We  present a high-level overview of the proof in \cref{appendix:highlevel}, detailing the claims we have adapted and the general proof strategy. For completeness, we {also present complete proofs of all claims} in \cref{appendix:lex-hatfield}.}

\subsection{Partition under Strict Preferences}\label{sec:strict-partition} 
In this section we extend the main result for partition mechanisms, from the case that the preferences are lexicographic, to the case that the class of preferences is any permutations-closed class of strict preferences that contains the lexicographic preferences, that is, to any class $\PP$ in $ \FFn(\MM)$.

An immediate consequence of \cref{proposition:lexi-case} is that partition mechanisms over ``general'' preference classes $\PP\in \FFn(\MM)$ (i.e. classes containing the lexicographic preferences)
which are truthful, non-bossy, and neutral also behave like serial-quota mechanisms for all \emph{lexicographic preference profiles}. Formally:

\begin{claim}
    \label{lexi-serial-quota}
    {Let $f$ be a truthful, non-bossy, and neutral partition mechanism of a set $\MM$ of $m$ goods among a set $\NN$ of $n$ agents with strict preferences from some class $\PP\in \FFn(\MM)$.}  Then there exists a quota-vector and an ordering of the agents $\qp\in \Qnm$ s.t. $f$ is a $\qp$-serial-quota mechanism on $\Plex(\MM)$.
    
\end{claim}

\begin{proof}
    Consider the partition mechanism $f^L$ that is the restriction of $f$ to lexicographic preferences. That is, $f^L$ is the partition mechanism of the set $\MM$ of $m$ goods among the set $\NN$ of $n$ agents  with lexicographic preferences $\Plex(\MM)$ defined identically to $f$:    \[f^L\veco = f\veco \ \ \ \ \forall \veco\in\Plex^n(\MM)\]
    {Note that $f^L$ is well defined since the assumption $\PP\in \FFn(\MM)$ requires that $\Plex^n(\MM)\subseteq\PP$.}
    Since $f$ is truthful, non-bossy, and neutral, $f^L$ must also be truthful, non-bossy, and neutral. Then \cref{proposition:lexi-case} implies that there exists a quota-vector and ordering of the agents $\qp\in\Qnm$ such that $f^L$ is a $\qp$-serial-quota mechanism.     Then for all lexicographic preference profiles $\veco\in \Plex^n(\MM)$:   
    \[f\veco=f^L\veco=\SQqp\veco\]
    as required.
\end{proof}

In order to prove our main result for partition mechanisms over classes of strict preferences (\cref{strict-partition}), we must show that the above equality holds for all preference profiles over $\PP^n$, not just the lexicographic ones. First, we prove the following useful lemma.

\begin{restatable}{lemma}{constantquotas}
    
\label{constant-quotas}
    {Let $f$ be a truthful, non-bossy, and neutral partition mechanism of a set of goods  $\MM$  among $n$ agents with strict preferences from some class    $\PP\in \FFn(\MM)$.} Let $\qp\in \Qnm$ be {the quota and ordering for which $f$ is a $\qp$-serial-quota mechanism on \emph{lexicographic} preferences (by \cref{proposition:lexi-case})}.  Then for every preference profile $\veco\in \PP^n$ {(not just lexicographic profiles)} and index $i\in[n]$:  \[|f_{p_i}\veco|=q_i\]
\end{restatable}

\begin{proof} Let $\qp$ be as defined in \cref{lexi-serial-quota}. Let $\veco\in \PP^n$  be some preference profile. For each agent $i\in\NN$ let $\preqa_i\in \Plex(\MM)$ be some lexicographic preference according to an ordering of $\MM$ that places the goods of $A_i=f_i\veco$ first. {That is, the ordering is consistent with $(A_i,\MM\setminus A_i)$}. Consider the profile of lexicographic preferences $\vecoa$. \cref{lexi-serial-quota} implies that for every index $i\in [n]$:

\[f_{p_i}\vecoa=\SQqp_{p_i}\vecoa=\dem{\preqa_{p_i}}{q_i}{\MM\setminus\left( \bigcup_{j=1}^{i-1} \SQqp_{p_j}\vecoa \right)}\]
Since preferences are strict and monotone we know that for each index $i\in[n]$  agent $p_i$ will always demand a set of size equal to their quota $q_i$. Then $|f_{p_i}\vecoa|=q_i$. On the other hand, clearly $\vecoa$ is a push-up of $\veco$ for $f$. As $f$ is a truthful and non-bossy allocation mechanism, \cref{monotonicity} implies that $f\veco=f\vecoa$. Then:
\[|f_{p_i}\veco|=|f_{p_i}\vecoa|=q_{i}\]
as required.
\end{proof}

This result already shows that each agent must have some a-priori ``quota'' for exactly how many goods they receive from $f$ for all preference profiles. We will now extend this to show that the agents select their goods according to their quota {in a serial fashion}. 

\subsubsection{Strong Desire and Control}
\label{subsec:strong-desire-and-control}

In order to generalize the results from the lexicographic domain to all 
preference classes in $\FFn(\MM)$ (classes of strict preferences that contain the lexicographic preferences) we make use of the  concepts of \emph{strong desire} and \emph{control}. The concepts are adapted from Amanatidis et al. \cite{AmanatidisBCM17}, and altered to address weakly-monotone preferences in general (rather than only additive valuations). 

{An agent \emph{strongly desires} $S$ if she always prefers getting even one more good from $S$, over additionally getting the entirety of $\MM\setminus S$. Formally:}
\begin{definition}
Agent with preference $\preq$ \emph{strongly desires} a set $S\subseteq \MM$ if for every $A \subsetneq B \subseteq S$:
\[ A \cup (\MM \setminus S) \prec B\]
\end{definition}

For example, for an ordering $e_1,\ldots ,e_m$ of the goods, an agent with the corresponding lexicographic preference strongly desires every set of goods that forms a prefix of the ordering. That is, each of the following sets is strongly desired: $\{e_1\}, \{e_1,e_2\}, \{e_1,e_2,e_3\},\ldots, \{e_1,e_2,\ldots, e_n\}$.

An agent \emph{controls} a set $S\subseteq \MM$ with respect to $f$ if whenever she strongly desires $S$, she gets $S$ (and possibly some other goods). Formally:
\begin{definition}
    Let $f$ be an allocation mechanism for agents with preferences from class $\PP$. 
    We say that agent $i\in \NN$ \emph{controls a set $S\subseteq \MM$ with respect to $f$,} 
    if for every $\veco\in \PP^n$  in which agent $i$  with preference $\preq_i$ strongly desires $S$, it holds that $S \subseteq f_i\veco$.
\end{definition}
Notice that the notion of control is defined for a specific agent $i$ and set $S$, but constrains the allocation for many different preference profiles of $i$, so long as $i$ strongly desires $S$. 

\subsubsection{The Control Lemma} \label{subsec:control-lemma}
The cornerstone of the proof of our result for strict preferences is the Control Lemma. {The lemma presents conditions on a partition mechanism under which the ``control claim'' holds. 
Mechanism $f$ satisfies the \emph{control claim} if the following holds: If there exists a profile $\veco\in \PP^n$ such that all agents strongly desire $S\subseteq \MM$, and yet $S\subseteq f_i\veco$ for some  $i\in \NN$, then agent $i$ controls $S$ with respect to $f$.
The Control Lemma shows that the control claim holds for partition mechanism $f$ that satisfies three properties: it is truthful, non-bossy, and neutral.} {Formally:}

\control*

{
The general steps of the proof are as follows. We presume an instance of the control claim, that is, a profile $\veco$ where all agents strongly desire some set $S\subseteq \MM$ yet some agent $i\in\NN$ gets it whole. We wish to show that agent $i$ controls $S$, that is, for any profile $\vecoa$ where agent $i$ strongly desires $S$ she also gets it whole. We construct lexicographic preference profiles from $\veco$ and $\vecoa$, and prove that {their corresponding allocations are equal} (\cref{control-c1} and \cref{control-c2}).
Finally we use \cref{lexi-serial-quota}, which states that $f$ behaves like a serial-quota mechanism on lexicographic preferences, to show that agent $i$ getting the entirety of $S$ in the lexicographic profile corresponding to $\veco$ implies the same for the lexicographic profile corresponding to $\vecoa$.

We now present the formal proof of this Lemma:}

{
    \begin{proof}[Proof of \cref{control} (Control Lemma)]
    Let $\veco\in \PP^n$ be a profile such that all agents strongly desire $S$, and yet $S\subseteq f_i\veco$ for some  $i\in \NN$. Let $\vecoa\in \PP^n$ be a profile for which agent $i$ strongly desires $S$. 
    To show that agent $i$ controls $S$ with respect to $f$ we prove that $S\subseteq f_i\vecoa$.
    
    For each agent $j\in \NN$ let $A_j=f_j\veco$ denote that {agent's} allocation, $R_j=A_j\setminus S$ denote the goods that $j$ got but are not in $S$,
    and $R_{-j}=\MM\setminus (R_j\cup S)$ denote the remaining goods. Then:
    \[f\veco=(A_1,\ldots , A_n)=(R_1,\ldots R_{i-1}, R_i\cup S,R_{i+1}, \ldots ,  R_n)\]
    For each agent $j\in \NN$ let $\vecp{e}_j$ be an ordering of the goods of $\MM$ that is consistent with $(S,R_j,R_{-j})$.
    
    We will now consider the profile $\vecoa$. For each agent $j\in \NN$ let $A'_j=f_j\vecoa$ denote the allocation of $j$ under $\vecoa$. 
    We will define an ordering $\vecp{e'}_j$ for each agent $j\in \NN$ as follows.
    \begin{itemize}
        \item For agents $j\in\NN\setminus \{i\}$ let $\vecp{e'}_j$ be an ordering of the goods that is consistent with $(A'_j,\MM\setminus A'_j)$.
        \item For agent $i$, first partition the set $S$ into $S_i=S\cap A'_i$ and $S_{-i}=S\setminus A'_i$, and define $R'_i=A'_i\setminus S$ and $R'_{-i}=\MM\setminus(R'_i\cup S)$.     Now define an ordering $\vecp{e'_i}$ which is consistent with $(S_i,S_{-i},R'_i, R'_{-i})$.
    \end{itemize}
    
    Consider the lexicographic preference profiles $\vecordp{e}, \vecordp{e'}$ corresponding to these two different orderings of the goods for each agent. We will prove two claims:
    \begin{enumerate}
        \item $\vecordp{e}$ is a push-up of $\veco$ for $f$
        \item $\vecordp{e'}$ is a push-up of $\vecoa$ for $f$
    \end{enumerate}
    The proofs of these claims are rather technical, {and presented {in \cref{appendix:control}} as \cref{control-c1} and \cref{control-c2} respectively.} We will show that these claims together complete the proof. Since $f$ is a truthful and non-bossy partition mechanism, \cref{monotonicity} (monotonicity) and the above two claims imply:
    \[f\vecordp{e}=f\veco=(A_1,\ldots , A_n)\]
    \[f\vecordp{e'}=f\vecoa=(A'_1,\ldots , A'_n)\]
    From the original assumption and the first equation we know $S\subseteq f_i\veco=f_i\vecordp{e}$. The second equation shows that proving $S\subseteq f_i\vecoa$ is equivalent to proving  $S\subseteq f_i\vecordp{e'}$.
    
    Since $f$ is truthful, non-bossy, neutral and partitions the goods, \cref{lexi-serial-quota} implies there exists a quota-vector and ordering of the agents $\qp\in\Qnm$ such that $f=\SQqp$ for all lexicographic preference profiles. Note that $\vecordp{e}, \vecordp{e'}$ are both lexicographic preference profiles.
    
    Start with the profile $\vecordp{e}$. All agents lexicographically place the goods of $S$ the highest {(they strongly desire $S$)} yet $S\subseteq f_i\vecordp{e}=\SQqp_i\vecordp{e}$. This is only possible if agent $i$ chooses first with a quota that is larger than $S$, meaning both  $p_1=i$ and $|S|\leq q_1$. Now consider the profile $\vecordp{e'}$. Since the same $\qp$ must apply for all lexicographic preference profiles, over this profile as well agent $i$ will also choose first and take $q_1$ goods. Since the preference $ \preq_{\vecp{e'_i}}$ also places the goods of $S$ the highest, she will again pick them first and get them all. Therefore:
    \[S\subseteq \SQqp_i\vecordp{e'}=f_i\vecordp{e'}\]
    {as required}. 
    \end{proof}
}

We now
{present} some implications of the Control Lemma which will be of great use in the proof of the main result for partition mechanisms over classes of strict preferences, presented in \cref{subsec:strict-partition}.

\begin{restatable}{lemma}{firstcontrol}
    \label{first-control}
    Let $f$ be a truthful, non-bossy, and neutral mechanism for agents with strict preferences from some class $\PP\in \FFn(\MM)$. Let $\qp\in \Qnm$ be {the quota and ordering for which $f$ is a $\qp$-serial-quota mechanism on \emph{lexicographic} preferences (by \cref{proposition:lexi-case})}.

    {Then agent $p_1$ controls the set $S\subseteq \MM$ with respect to  $f$ if and only if $|S|\leq q_1$.}
\end{restatable}

\begin{proof} We will prove each direction separately.
\textit{Direction 1:} Let $S\subseteq \MM$ such that  $|S|\leq q_1$, we show that agent $p_1$ controls $S$ for $f$. Let $\veco\in \Plex^n(\MM)$ be a lexicographic preference profile such that all agents place the goods of $S$ the highest (in some arbitrary order). Clearly  $S\subseteq \dem{\preq_{p_1}}{q_1}{\MM}$, therefore $S\subseteq \SQqp_{p_1}\veco$ {(since agent $p_1$ has quota $q_1\geq |S|$, she can pick all the goods in $S$)}. Then by \cref{proposition:lexi-case}:
\[f_{{p_1}}\veco=\SQqp_{p_1}\veco\supseteq S\]
We have shown that over a profile where all agents strongly desire $S$ it holds that $S\subseteq f_{{p_1}}\veco$ therefore by \cref{control} (control) agent ${p_1}$ controls $S$ for $f$, as required.

\textit{Direction 2:} Let $S\subseteq \MM$ such that  $|S| > q_1$, assume for contradiction that agent ${p_1}$ controls $S$ for $f$. Let $\veco\in \PP^n$ be some preference profile such that agent $p_1$ strongly desires $S$, then since she controls $S$ for $f$ this implies $S\subseteq f_{{p_1}}\veco$. But then $|f_{{p_1}}\veco|\geq |S| > q_1$, contradicting \cref{constant-quotas}.
\end{proof}

\begin{restatable}{lemma}{firstpicker}
     \label{first-picker}
  {For the same $f,\PP$, and $\qp$ as in \cref{first-control}, it holds that} agent $p_1$ always gets her $q_1$-demand from $\MM$. Formally, for every profile $\veco\in \PP^n$, $f_{p_1}\veco=\dem{\preq_{p_1}}{q_1}{\MM}$.
\end{restatable}

\begin{proof}
     Let $\veco\in \PP^n$ be some preference profile. \cref{constant-quotas} implies $|f_{p_1}\veco|=q_1$. Then by definition of demand, $f_{p_1}\veco \preq_{p_1} \dem{\preq_{p_1}}{q_1}{\MM}$. Assume for contradiction that $f_{p_1}\veco\neq\dem{\preq_{p_1}}{q_1}{\MM}$, then since $\preq_{p_1}$ is strict $f_{p_1}\veco \prec_{p_1} \dem{\preq_{p_1}}{q_1}{\MM}$. Let $\preqa_{p_1}\in \PP$ be an alternate preference profile s.t. agent $p_1$ with preference $\preqa_{p_1}$ strongly desires $\dem{\preq_{p_1}}{q_1}{\MM}$. Since $|\dem{\preq_{p_1}}{q_1}{\MM}|=q_1$, \cref{first-control} implies that agent $p_1$ controls $\dem{\preq_{p_1}}{q_1}{\MM}$. Then by the definition of control $\dem{\preq_{p_1}}{q_1}{\MM}\subseteq f_{p_1}(\preqa_{p_1},\preq_{-p_1})$. Thus:
\[ f_{p_1}\veco \prec_{p_1} \dem{\preq_{p_1}}{q_1}{\MM} \preq_{p_1} f_{p_1}(\preqa_{p_1},\preq_{-p_1})\]
But this means that when the profile is $\veco$, agent $p_1$ with true preference $\preq_{p_1}$ would benefit from lying and reporting $\preqa_{p_1}$ instead, contradicting the assumption that $f$ is truthful.
\end{proof}

To complete the picture, we show that the Control Lemma is tight in its requirements from the mechanism, in the sense that for $n\geq 3$ agents, dropping any of the three properties (truthful, neutral or non-bossy) will break it\footnote{For two agents, Amanatidis et. al \cite{AmanatidisBCM17} have proven the control claim for truthful partition mechanisms. Such mechanisms are necessarily non-bossy, so for $n=2$ the control claim holds even without neutrality. }. In fact, we show that even for $n\geq2$ agents, once truthfulness is dropped the control claim does not hold.  We present a proof of the following claim in \cref{appendix:control-nec}.

\begin{restatable}{claim}{controltight}
\label{lemma:control-tight}

{For $n\geq 3$ agents, for \emph{every} class of strict preferences $\PP\in\FFn(\MM)$, 
for each of the three properties of \emph{truthfulness, non-bossiness, and neutrality},
there exists a partition mechanism $f$ which violates only that property, such that the control claim fails for $f$.}

\end{restatable}

\subsubsection{Proving the Strict Version}\label{subsec:strict-partition}

The entirety of the proof thus far has resulted in proving two key results:
\begin{enumerate}
    \item As consequence of \cref{proposition:lexi-case}, any truthful, non-bossy and neutral mechanism behaves like a serial-quota mechanism for lexicographic preference profiles\footnote{{See \cref{lexi-serial-quota} for the formal claim.}}.
    \item \cref{first-picker} shows that for every truthful, non-bossy and neutral mechanism there is an agent $p_1$ with some quota $q_1$ such that $p_1$ always gets her $q_1$-demand from $\MM$.
\end{enumerate}

In this section, we use these results to prove the following characterization for partition mechanisms: 

\begin{restatable}{lemma}{strictpartition}
\label{strict-partition}
    {Let $f$ be a partition mechanism of a set $\MM$ of $m$ goods among a set $\NN$ of $n$ agents with strict preferences from some class $\PP\in \FFn(\MM)$.} Mechanism $f$ is truthful, non-bossy, and neutral if and only if $f$ is a serial-quota partition mechanism.
\end{restatable}

{
The proof of the characterization is highly technical, 
{so we next only present}
an intuitive high-level overview of the proof, which captures the core ideas behind this characterization. The formal proof of the characterization, along with the necessary definitions and intermediate results is deferred to \cref{appendix:strict-partition}. 
}

{The heart of the proof, is based on an induction statement: 
First, we use \cref{first-picker} to show there must exist a ``first agent'' $p_1$ with some quota $q_1$ that always gets her $q_1$-demand from $\MM$.  Next, we prove that after the first agent $p_1$ has taken her demand, the mechanism that allocates the remaining goods to the remaining agents is also truthful, non-bossy, and neutral (see \cref{first-pick-independence} and \cref{induced-mechanism}). Applying the induction hypothesis, this tells us the induced mechanism is a serial-quota mechanism, for some quotas and ordering of the remaining agents. We then must show that the ordering and quotas of the induced mechanism are the same as the ordering and quotas of the original mechanism $f$ over lexicographic preference profiles. 
Finally, we conclude that $f$ is a serial-quota partition mechanism for all preference profiles.}

\subsection{Main Characterization Result: Allocation under Strict Preferences}\label{sec:strict-alloc}
In this section we 
extend the characterization result that was proven in \cref{sec:strict-partition} for partition mechanisms (\cref{strict-partition}), to hold also for allocation mechanisms (that do not have to allocate all goods). In essence, the claim follows as allocating goods to $n$ agents (with some goods left unallocated) is like partitioning the goods to $n+1$ agents (giving the extra agent the unallocated goods). {We present the proof in \cref{appendix:alloc-strict}.}

\strict*

\section{Implications for Cardinal Preferences (Valuation Functions)}\label{sec:implications-val}

In this section we move to consider agents with \emph{cardinal preferences} (also known as \emph{valuations}), and 
use our characterization result for ordinal preferences to derive characterization results for classes of valuations functions.

\begin{definition}[{Valuations}]
    A \emph{valuation} (or \emph{valuation function} or \emph{cardinal preference}) over a set  $\MM$ of goods  is a function $v:2^{\MM}\rightarrow \mathbb{R}_{\geq 0}$. 
    The (cardinal) \emph{utility (or value)} of an agent with valuation $v$ when getting the set $S$ is $v(S)$.
    Each valuation $v$ induces an ordinal preference, denoted $\preq_v$, as follows:
    \[\forall S,T\subseteq M: S\preq_v T \iff v(S)\leq v(T)\]
For a property $P$ of preferences, we say that a valuation $v$ has property $P$ if its induced preference $\preq_v$ has property $P$. For example, we say that a valuation function $v$ is \emph{strict} if the preference $\preq_v$ is strict. 

We denote the set of all \emph{(weakly) monotone} valuations over $\MM$ by $\Vmon(\MM)$ and only consider preferences in that set.
A \emph{valuation class}  $\VV$ is a set of valuation functions such that $\VV\subseteq \Vmon(\MM)$. A \emph{valuation vector} or \emph{valuation profile} of $n$ agents over $\VV$ is a vector $\vecv\in \VV^n$. The \emph{induced preference class} of the valuation class $\VV$, is the class of ordinal preferences $\PP_{\VV}=\{\preq_v:v\in \VV\}$.
\end{definition}

{
For the sake of formality, in \cref{appendix:cardinal-allocations} we re-define all the relevant definitions from ordinal preferences to cardinal valuations.
}
{The following lemma shows that {when valuations are strict, the output of a} truthful and non-bossy cardinal allocation mechanism depends {only} on the induced preferences.} The proof can be found in \cref{appendix:subsec:implications-val}.
\begin{restatable}{lemma}{setOrderingEquivalence}
    \label{setOrderingEquivalence}
    For valuation class $\VV$ of strict valuation functions, let $f:\VV^n\rightarrow\ALLOCsM$ be a truthful and non-bossy allocation mechanism. Then for every pair $\vecv,\vecu\in \VV^n$ of valuation vectors s.t. $\preq_{v_i}=\preq_{u_i}$ {for every agent $i\in\NN$,} it holds that $f(\vecv)=f(\vecu)$.

\end{restatable}
{
We now use our characterization for ordinal preferences to derive a characterization result for mechanisms over strict cardinal valuations. The proof can be found in \cref{appendix:subsec:implications-val}.
}

\cardinalstrict*

\subsection{Classes with some Monotone Valuation Functions}

{When the class $\VV$ contains valuations which are not strict, the properties defined above are no longer well defined. In order to extend the implications of our characterization to valuations which are not strict, we will need to 
consider the restriction of a class of valuations to the sub-class of strict valuations within it.
\begin{definition}
    Let $\VV\subseteq \Vmon(\MM)$ be a class of weakly-monotone valuation functions. We denote the intersection of $\VV$ and strict valuation functions by $\VV_{\neq}$.
\end{definition}
}
For example, if $\VV$ is the class of all additive valuations over goods, the class $\VV_{\neq}$ is the class of additive valuations over goods for which two different subsets of goods never have the same utility (no ties). We say a mechanism $f$ satisfies property $K$ on sub-class $\VV_{\neq}$ if the restriction of $f$ to only accept strict valuations from $\VV_{\neq}$ satisfies $K$. For example, we say a mechanism $f$ is \emph{neutral on $\VV_{\neq}$} if for every profile $\veco\in \VV_{\neq}^n$ and permutation $\pi\in \perM$, $f\vecopi=\pi(f\veco)$.

For a class of valuations that includes valuations with ties, the next corollary shows that we know how every cardinal allocation mechanism that is truthful, non-bossy, and neutral must behave on the sub-class of \emph{strict} valuations: it must behave like a serial-quota mechanism on these valuations. For example, such a mechanism must behave like a serial-quota mechanism on the sub-class of additive valuations that are strict. 

\cardinalnonstrict*
        
\section{Implications for Fair Division}\label{sec:implications-fair}

    In this section we discuss the implications of our characterization result to the problem of fair allocation of indivisible goods when agents have additive valuations over goods. We present strong negative results for both MMS approximation and {the envy-based notion of EF1}.

    As usual, we say that valuation $v$ is \emph{additive over goods} if $v(\{x\})\geq 0$ for every good $x\in \MM$, and for every set $S\subseteq\MM$ it holds that $v(S)=\sum_{x\in S}v(\{x\})$. We next present the fairness notations and derive the implications of our characterization results.''

\subsection{The Maximin Share (MMS)}

\begin{definition}[{MMS, $\rho$-MMS}]
For every $n\in\mathbb{N}$ and valuation $v:2^\MM\rightarrow\mathbb{R}$, let MMS$(v,n)$ denote the maximin-share of an agent with valuation $v$ in a mechanism with $n$ agents. Formally:

\[\text{MMS}(v,n)=\max_{\left(S_{1},\ldots, S_{n}\right)\in\text{Par}_{n}(\MM)}\min_{i\in\NN}\left\{ v(S_{i})\right\}\]
{Where $\text{Par}_n(\MM)$ denotes the family of all partitions of the set of goods $\MM$ into $n$ sets.}

We say a cardinal allocation mechanism $f$ for $n$ agents with valuations from class $\VV$ is $\rho$-MMS if for every valuation vector $\vecv\in \VV^n$ and agent $i\in \NN$:
\[v_i(f_i(\vecv))\geq \rho\cdot \text{MMS}(v_i,n)\]
\end{definition}

It was observed in \cite{AmanatidisBM16} that all serial-quota mechanisms have a poor MMS approximation (Theorem 4.1). We restate their result:

\begin{claim}
    \label{approximation-bound}
    For $n\geq 2$ agents with strict additive valuations over $m\geq n+2$ goods, 
    the serial-quota mechanism with quota vector $\vecp{q}=(1,1,1,\ldots,1,1, m-n+1)$ 
    is $\rho$-MMS for $\rho=\frac{1}{\lfloor{\frac{m-n+2}{2}}\rfloor}$, and for any larger constant $\rho$ there is no serial-quota mechanism that is is $\rho$-MMS.
\end{claim}

When the number of goods is even mildly larger than the number of agents (say $m\geq 2n$), the above bound is $O(1/m)$, essentially the poorest MMS approximation imaginable (as just letting each agent pick one good in turn clearly gives an approximation of $\Omega(1/m)$).
With our characterization, we can derive that the poor MMS approximation of serial-quota mechanisms is actually the best one can hope from any  truthful, non-bossy, and neutral mechanism.

\mms*
Notice that the class of additive valuations is permutations-closed and induces the lexicographic preferences. Therefore the above application follows immediately from \cref{corollary:cardinal-nonstrict} and \cref{approximation-bound}.

{The simple mechanism {for strict additive valuations} in which each agent picks one good in her turn {and the last agent receives all remaining goods} is $\rho$-MMS for $\rho=\frac{1}{\lfloor{\frac{m-n+2}{2}}\rfloor}$, and this mechanism is a serial-quota mechanism so it is  truthful, non-bossy, and neutral on strict additive preferences.  Thus, this result shows that no {mechanism which is} truthful, non-bossy, and neutral {on strict additive valuations} can have a better MMS approximation than this trivial mechanism.}

{We thus see that the impact of truthfulness on performance is significant: in contrast to our negative result showing that $\rho$ is at most $O(1/m)$ for truthful mechanisms, {allocation algorithms (that are not truthful but are neutral and non-bossy) give $\rho$ that is a constant. 
For example, $\rho=1/2$ can be obtained by} the ``bag-filling'' algorithm presented in \cite{Garg2019}. This algorithm is a simple greedy allocation algorithm (that is non-truthful) which is non-bossy and neutral.}

\subsection{Envy-based Notions}
We next discuss a central envy-based fairness notion for indivisible goods, the notion of \emph{envy-free-up-to-one-good (EF1)}. This notion is fundamentally ordinal, so we present its definition as an ordinal notion. 

\begin{definition}[{EF1}] 
Fix a class $\PP=\PP(\MM)\subseteq \Pnmon(\MM)$ of strict preferences over a set of goods $\MM$.

    An allocation $A=(A_1,...,A_n)$ is \emph{EF1} for $n$ agents with ordinal preferences $\veco\in \PP^n$, if for every pair of agents $i\neq j$ \emph{there exists a good} $x\in A_j$ such that $A_j \setminus\{x\} \preq_i A_i$.
    
    \end{definition}

EF1 says that envy can be eliminated by removing \emph{one} good from the other agent's set. The EF1 property is defined to hold for cardinal preferences (valuations) if the property holds for the induced preferences. 
Finally, we say that a mechanism is EF1 if it always outputs an EF1 allocation with respect to the input preferences.

{Any allocation that gives all agents at most one good is trivially EF1, so the serial-quota mechanism which gives each agent a quota of at most one is EF1. When all agents are given a quota of exactly one, the last agent will definitely get a good which all agents value less (or equal) to the good that they got. Then we can actually let the last agent select two goods instead of one, and since all the previous agents will prefer their own good over each of these two goods individually, the allocation will still be EF1. In the following claim, we show that there are no other possible quotas for which a serial-quota mechanism is EF1. {The proof is deferred to \cref{appendix:implications-fair}.}

\begin{restatable}{claim}{quotaefone}
\label{quota-efone}
Consider $n\geq 2$ agents with strict additive valuations over $m\geq n+1$ goods\footnote{When $m\leq n$ there is a trivial positive result: the {$n$-agent} Round Robin mechanism {for allocating $m\leq n$ goods} (which is a serial-quota mechanism) is EF1 and also truthful, non-bossy, and neutral.
{Note the truthfulness hinges on the assumption that $m\leq n$, and breaks for large $m$}.}.
    Any  serial-quota mechanism that is EF1     cannot have quota larger than 1 for any of the agents but the last, and {the last} cannot have quota larger than 2 (and the quota of the last is at most 1, unless all other agents have quota {of exactly 1}).
    Thus, the serial-quota mechanism that allocates the highest number of goods has quota vector $\vecp{q}=(1,1,1,\ldots,1,1,2)$.
\end{restatable}

 We thus see that a serial-quota mechanism that is EF1 can have only a tiny advantage over the trivial EF1, truthful, non-bossy, and neutral mechanism that simply allows each agent to pick one good in her turn: it is possible to give the last agent the opportunity to pick one more good (but nothing more). Combining   \cref{quota-efone} with the characterization from \cref{corollary:cardinal-nonstrict} immediately implies the following result:

\addef*

We thus see that the impact of truthfulness on performance is significant: our result shows that any truthful,  neutral and non-bossy mechanism that is EF1 must throw away almost all goods (throw at least $m-(n+1)$ goods), while the Round Robin mechanism is neutral and non-bossy, and always outputs an EF1 partition of \emph{all} goods.

\section*{Acknowledgments}
This research was supported by the Israel Science Foundation (grant No. 301/24).
Moshe Babaioff's research is also supported by a Golda Meir Fellowship.

\expandafter\def\csname blx@maxbibnames\endcsname{99}\printbibliography

\pagebreak
\appendix

\section{Additional Discussion of the Model}

\subsection{Lexicographic Preferences}
\label{appendix:subsec:lexi}

In this section we provide some more information about lexicographic preferences that will be used in the proofs. See \cref{def:lexi} for a formal definition of lexicographic preferences. An important characteristic of lexicographic preferences that is useful is that all lexicographic preferences are permutations of each other. 

\begin{claim}
    Let $\preq, \preqa\in \Plex(\MM)$ be lexicographic preferences over a set of goods $\MM$. Then there exists a permutation $\pi\in\perM$ such that $\preqa=\preq^{\pi}$.
\end{claim}
\begin{proof}
    Let $\vecp{e}=(e_1,\ldots, e_m)$ and $\vecp{e'}=(e'_1,\ldots, e'_m)$ be the orderings of the goods corresponding to $\preq,\preqa$ respectively, so that $\preq=\preq_{\vecp{e}}$ and $\preqa=\preq_{\vecp{e'}}$. Consider the following permutation $\pi\in \perM$ {which} switches between the two {orderings}: the permutation defined by $\pi(e_j)=e'_j$ for all $j\in[m]$. We show that this same permutation also swaps between the corresponding lexicographic preferences, so that $\preq^{\pi}=\preqa$.
    
    Let $S,T\subseteq\MM$ be two subsets. By definition, $T\preq^{\pi} S$ hold if and only if there exists $j\in[m]$ such that $\pi(e_j)=e'_j\in S\setminus T$ and for all $1\leq k\leq j$, $\pi(e_k)=e'_k\in S\iff \pi(e_k)=e'_k\in T$. Notice this is equivalent to $T\preqa S$.
\end{proof}

{
Thus, there is a one-to-one correspondence between permutations and lexicographic preferences: from any given lexicographic preference, each of its permutations is a distinct lexicographic preference, and any such preference is one of its permutations. Thus the set of all its permutations is exactly the set of all lexicographic preferences. Note this is a stronger property than just being permutations-closed.
}

\subsection{Serial-Quota Mechanisms}
\label{appendix:subsec:sq-mech}

In this section we further explain the motivation behind the exact definition of the set $\Qnm$ of all pairs of valid quota-vectors and ordering presented in \cref{quota-ordering-pair}.

For the simplicity of some proofs (for example, \cref{strict-partition}) it is useful if each serial-quota mechanism is represented by a single (unique) quota-ordering pair $\qp\in\Qnm$. An issue that may arise here is that when some agents have a quota of $0$, changing their position in the ordering does not change the outcome of the mechanism in any way (but have different representations as vectors $\qp$). For this reason, we defined the set of pairs $\Qnm$ so that when there are agents with a quota of zero, they are always placed last in the ordering, ordered internally according to their indices. These requirements guarantee that if an allocation mechanism $f$ is a serial-quota mechanism, there exists a \emph{unique} $\qp\in\Qnm$ for which $f$ is a $\qp$-serial-quota mechanism.

\section{\texorpdfstring{Missing Proofs from \cref{sec:monotonicity}}
                        {Missing Proofs from Section~\ref{sec:monotonicity}}}

\label{appendix:monotonicity}

\monotonicity*

\begin{proof}
We prove the Lemma in two steps, first proving the claim holds if the preference of only one agent is pushed up, and then using induction to prove the claim when multiple preferences are pushed up. 

\emph{Step 1:} Let $\veco, \vecoa\in \PP^n$  be a pair preference profiles such that $\vecoa$ is a push-up of $\veco$ for $f$. We prove that for any $i\in \NN$ it holds that $f\veco=f(\preqa_i,\preq_{-i})$. Fix some  $i\in \NN$. As $f$ is  truthful we know that $f_i(\preqa_i,\preq_{-i})\preq_i f_i\veco$. As $\vecoa$ is a push-up of $\veco$, this implies $f_i(\preqa_i,\preq_{-i})\preqa_{i} f_i\veco$. Truthfulness of $f$ in the opposite direction implies $f_i\veco\preqa_{i} f_i(\preqa_i,\preq_{-i})$, so together these imply $f_i\veco= f_i(\preqa_i,\preq_{-i})$ (since $\preqa_i$ is strict). Then non-bossiness implies $f\veco=f(\preqa_i,\preq_{-i})$, as required.

\emph{Step 2:} Let $\veco, \vecoa\in \PP^n$  be a pair preference profiles such that $\vecoa$ is a push-up of $\veco$ for $f$. We will prove by induction on $i$ from $0$ to $n$ that:
\[f(\preqa_1,\ldots, \preqa_i,\preq_{i+1},\ldots, \preq_n)=f\veco\]
The base case $i=0$ is trivial. Let $i>0$, we assume the induction hypothesis for $i-1$:
\[f(\preqa_1,\ldots, \preqa_{i-1},\preq_{i},\ldots, \preq_n)=f\veco\]

We will prove the induction statement for $i$. First we show that the profile $\vecoa$ is a push-up of $(\preqa_1,\ldots, \preqa_{i-1},\preq_{i},\ldots, \preq_n)$ for $f$. For agents $j\leq i-1$ this is trivial since their preferences over both profiles are identical. We must show that for each agent $j\geq i$, the preference $\preqa_j$ is a push-up of $\preq_j$ for allocation $f_j(\preqa_1,\ldots, \preqa_{i-1},\preq_{i},\ldots, \preq_n)=f_j\veco$. Since we assumed that $\vecoa$ is a push-up of $\veco$ for $f$ we know that $\preqa_j$ is a push-up of $\preq_j$ for allocation $f_j\veco$, as required.

We have shown that $\vecoa$ is a push-up of $(\preqa_1,\ldots, \preqa_{i-1},\preq_{i},\ldots, \preq_n)$ for $f$. Notice that the profile  $(\preqa_1,\ldots, \preqa_i,\preq_{i+1},\ldots, \preq_n)$ is obtained from the profile $(\preqa_1,\ldots, \preqa_{i-1},\preq_{i},\ldots, \preq_n)$ by having agent $i$ alter her preference from $\preq_i$ to $\preqa_i$. Applying the result from step $1$ we get that:
\[f(\preqa_1,\ldots, \preqa_{i},\preq_{i+1},\ldots, \preq_n)=f(\preqa_1,\ldots, \preqa_{i-1},\preq_{i},\ldots, \preq_n)=f\veco\]
Completing the induction. Taking the end of the induction $i=n$ we get $f\vecoa=f\veco$, completing the proof.
\end{proof}

We note that while truthfulness ensures that the allocation of agent $i$ does not change when the preference of $i$ changes to a push-up of her original preference, the fact that the allocation of every other agent does not change crucially depends on the mechanism being non-bossy.

\section{\texorpdfstring{Missing Proofs from \cref{subsec:lexi}}
                        {Missing Proofs from Section~\ref{subsec:lexi}}}

\subsection{\texorpdfstring{Missing Proofs from \cref{sec:lex-PO}}
                        {Missing Proofs from Section~\ref{sec:lex-PO}}}
\label{appendix:lexi-po}

{
The proof of \cref{proposition:Pareto-efficient} centers around the following novel definition of a \emph{strictly improving single-good trading cycle}.
}

\begin{definition}
    {A} \emph{single-good trading cycle} over an allocation $A=(A_1,\ldots, A_n)$ and a preference profile $\veco$, is an ordered list $\vecp{x}=(x_{c_1},\ldots, x_{c_k})$ of $k>1$ distinct goods, each allocated to a different agent. That is, for agent $c_j$ it holds that  $x_{c_j}\in A_{c_j}$, and $c_j\neq c_i$ for $j\neq i$.

    The \emph{augmented allocation $A^{\vecp{x}}$ using the trading cycle $\vecp{x}$} is the same as allocation $A$ except that for every index $j\in[k]$, agent $c_j$ is allocated good $x_{c_{j-1}}$ instead of good $x_{c_{j}}$, so that $A^{\vecp{x}}_{c_j}=A_{c_j}\setminus\{x_{c_j}\}\cup \{x_{c_{j-1}}\}$, using the cyclic notation $c_0=c_k$.
    
    We say the trading cycle $\vecp{x}$ is \emph{strictly improving} if for every $j\in[k]$, agent $c_j$ strictly prefers her augmented allocation, that is, $A_{c_j}\prec_{c_j} A^{\vecp{x}}_{c_j}$ for every index $j\in[k]$.
\end{definition}

{
Notice that the existence of a strictly improving single-good trading cycle for an allocation and preference profile implies that the allocation is not Pareto-optimal. Thus the non-existence of such cycles is a necessary condition for a mechanism to be Pareto-efficient.} In \cref{cycle-nonexistence} we  show that when preferences are lexicographic, for any truthful, non-bossy, and neutral allocation mechanism such a cycle cannot exist. Then, in \cref{proposition:Pareto-efficient} we use that claim to show that the non-existence of such cycles is sufficient to prove Pareto-efficiency when preferences are lexicographic. 

\begin{restatable}{lemma}{cyclenonexist}
    
    \label{cycle-nonexistence}
     Let $f$ be a truthful, non-bossy, and neutral allocation mechanism for $n$ agents with preferences from the class $\Plex(\MM)$ of lexicographic preferences. For every profile of lexicographic preferences $\veco\in \Plex^n(\MM)$ there cannot exist a strictly improving single-good trading cycle over the allocation $A=f\veco$.
\end{restatable}

\begin{proof} Assume for contradiction there exists a preference profile $\veco\in \Plex^n(\MM)$ s.t. there exists a strictly improving single-good trading cycle $\vecp{x}=(x_{c_1},\ldots, x_{c_k})$ (with $k>1$) over $\veco$ and the partition $A=f\veco$. 

We use the following notation to simplify the exposition. For each index $j\in [k]$ let $y_{c_j}=x_{c_{j-1}}$ be the good gained by agent $c_j$. Let $C=\{c_1,\ldots , c_k\}$ denote the set of all agents which are part of the trading cycle. For each agent $i\in C$ let $R_i= A_{i}\cap A^{\vecp{x}}_{i}=A_i\setminus \{x_i\}$ be the set of good that agent $i$ gets both before and after the augmentation, and let $R_{-i}=\MM\setminus(A_i\cup A^{\vecp{x}}_{i})=\MM\setminus(R_i\cup \{x_i,y_i\})$ be the set of goods that agent $i$ does not get neither before nor after the augmentation.

For each agent $i\in \NN$ we define an ordering of the goods $\vecp{e}_i$ as follows:
\begin{itemize}
    \item If $i\notin C$ let $\vecp{e}_i$ be any ordering that is consistent with $(A_i, \MM\setminus A_i )$.
    \item If $i\in C$ let  $\vecp{e}_i$ be any ordering that is consistent with $(R_i,y_i,x_i,R_{-i})$.
\end{itemize}
For ease of notation, for each agent $i\in\NN$ let $\preqb_i=\preq_{\vecp{e}_i}$ denote the lexicographic preference corresponding to the ordering $\vecp{e}_i$.

We also define the permutation of the goods that swaps the goods along the trading cycle. Formally, let $\pi\in \perM$ be the permutation of the goods defined $\pi(x_i)=y_i$ for every $i\in C$ and $\pi(z)=z$ for every $z\in\MM\setminus\vecp{x}$. Then for every agent $i\in \NN\setminus C$ not in the cycle, $\pi(A_i)=A_i=A^{\vecp{x}}_i $, while for every agent $i\in C$:\[\pi(A_i)=\pi(R_i\cup\{x_i\})=R_i\cup \{y_i\}=A^{\vecp{x}}_i \]

We will now prove two claims:
\begin{enumerate}
    \item $\vecob$ is a push-up of $\veco$ for $f$
    \item $\vecob$ is a push-up of $\vecoper{\preqb}{\pi}$ for $f$
\end{enumerate}

We will first show that these two claims together complete the proof. Since $f$ is a truthful and non-bossy  allocation mechanism, \cref{monotonicity} (monotonicity)\footnote{{See \cref{appendix:monotonicity}.}} along with the above two claims imply:
\[f\vecob=f\veco=A\]
\[f\vecoper{\preqb}{\pi}=f\vecob\]
Then combining these results with the neutrality of $f$ implies:
\[A=f\vecob=f\vecoper{\preqb}{\pi}=\pi\left(f\vecob\right)=\pi(A)=A^{\vecp{x}}\]
 Yet we know that for each agent $i\in C$, $A_i=R_i\cup \{x_i\}\neq R_i\cup \{y_i\}=A^{\vecp{x}}_i$, so this is a contradiction. Thus all that is left to complete the proof is to prove both claims.

\textit{Claim 1:} We prove that $\vecob$ is a push-up of $\veco$ for $f$. Let $i\in\NN$ and let $Z\subseteq \MM$ be a bundle such that $Z\preq_{i} f_i\veco=A_i$, we prove that $Z\preqb_iA_i$.

First consider the case when $i\notin C$. In this case we defined $\preqb_i$ to lexicographically place the goods of $A_i$ highest, so if it does not hold that $A_i\subseteq Z$ then immediately $Z\preqb_{i} A_i$, as required. Otherwise assume $A_i\subseteq Z$. Then $A_i \preq_i Z$, which combined with the assumption $Z\preq_i A_i$ and the strictness of $\preq_i$ implies $Z=A_i$, again proving $Z\preqb_{i} A_i$.

Now consider the case when $i\in C$. Recall then that $A_i=R_i\cup \{x_i\}$. If it does not hold that $R_i \subseteq Z$ then clearly $Z\preqb_i A_i$ because $R_i\subseteq A_i$ and $\preqb_i$ places the goods of $R_i$ highest. Otherwise, assume that $R_i \subseteq Z$. Since the cycle $\vecp{x}$ is strictly 
improving $A_{i}\prec_{i} A^{\vecp{x}}_{i}$. If it were the case that $y_i\in Z$ we would get $A^{\vecp{x}}_{i}=R_i\cup \{y_i\}\subseteq Z$, so $A^{\vecp{x}}_{i}\preq_i Z$. But then $Z\preq_iA_i\prec_iA^{\vecp{x}}_{i}\preq_iZ$, a contradiction. So we can assume that $y_i\notin Z$. If $x_i\in Z$ then $A_i=R_i\cup \{x_i\}\subseteq Z$. Then $A_i\preq_i Z$, which combined with the assumption $Z\preq_i A_i$ and the strictness of $\preq_i$ implies $Z=A_i$, again proving $Z\preqb_{i} A_i$. Otherwise assume  $x_i\notin Z$. Then $Z\subseteq R_i\cup R_{-i}$, and since $\preqb_i$ lexicographically prefers $x_i$ over the entirety of $R_{-i}$ this implies $Z\precb_iR_i\cup \{x_i\} =A_i$,again proving $Z\preqb_{i} A_i$ as required.

\textit{Claim 2:} We prove that $\vecob$ is a push-up of $\vecoper{\preqb}{\pi}$. Let $i\in\NN$ and let $Z\subseteq \MM$ be a bundle such that $Z\preqb^{\pi}_{i} f_i\vecoper{\preqb}{\pi}=\pi(A_i)=A^{\vecp{x}}_i $, we will prove that $Z\preqb_i A^{\vecp{x}}_i $.

First consider the case when $i\notin C$. In this case $ A^{\vecp{x}}_i=A_i$. If it does not hold that $A_i\subseteq Z$ then since $\preqb_i$ places the goods of $A_i$ the highest, this immediately implies $Z\preqb_i A_i$ as required. Otherwise, assume $A_i\subseteq Z$. Then $ A^{\vecp{x}}_i=A_i\preqb^{\pi}_{i} Z $, which combined with the earlier assumption that $Z\preqb^{\pi}_i  A^{\vecp{x}}_i$ and the strictness of lexicographic preferences implies $Z= A^{\vecp{x}}_i=A_i$, again proving $Z\preqb_i A^{\vecp{x}}_i $ as required.

Now consider the case when $i\in C$. Recall that in this case $A_i=R_i\cup \{x_i\}$,  and $ A^{\vecp{x}}_i=\pi(A_i)=R_i\cup \{y_i\}$.  If it does not hold that $R_i\subseteq Z$, then since $\preqb_i$ places the goods of $R_i$ the highest, this immediately implies $Z\preqb_i R_i\preqb_i A^{\vecp{x}}_i$ as required. Otherwise, assume $R_i\subseteq Z$. If also $y\in Z$ then $ A^{\vecp{x}}_i=R_i\cup\{y_i\}\subseteq Z$. Then $ A^{\vecp{x}}_i \preqb^{\pi}_{i}Z$ which combined with the earlier assumption that $Z\preqb^{\pi}_i  A^{\vecp{x}}_i$ and the strictness of $\preqb^{\pi}_i$ implies $Z= A^{\vecp{x}}_i$, again proving $Z\preqb_i A^{\vecp{x}}_i $ as required. Otherwise assume that $y_i\notin Z$. Then $Z\subseteq R_i\cup \{x_i\} \cup R_{-i}$. Since $\preqb_i$ lexicographically prefers $y_i$ to the entirety of $\{x_i\} \cup R_{-i}$: \[Z\preqb_i R_i\cup \{x_i\} \cup R_{-i} \precb_i R_i\cup \{y_i\} = A^{\vecp{x}}_i\]
Again proving $Z\preqb_i A^{\vecp{x}}_i$ as required.
\end{proof}

{Using this result, we can now prove \cref{proposition:Pareto-efficient}.}

\paretolexi*

\begin{proof} Assume for contradiction that $f$ is truthful, non-bossy, and neutral but it is not Pareto-efficient. Then there exists a lexicographic preference profile $\veco \in \Plex^n(\MM)$ such that the allocation $A=(A_1,\ldots, A_n)=f\veco$ is not Pareto-optimal. Meaning, there exists an alternate allocation $B=(B_1,\ldots, B_n)$ such that for all $i\in \NN: A_i\preq_i B_i$  and for at least one agent the inequality is strict. We will prove that the existence of a strictly improving single-good trading cycle over $\veco$ and $A$, thus contradicting \cref{cycle-nonexistence}.

Let $V$ be the set of agents for which the allocation of the agent in $A$ and in $B$ are different:
\[V = \{i\in \NN : A_i\neq B_i\}\]
Note that $V\neq \emptyset$ since the inequality $ A_i\preq_i B_i$ is strict for at least one agent $i\in \NN$. Also note that the entire class $\Plex(\MM)$ of lexicographic preferences is strict, therefore for every $i\in V$, $A_i\preq_i B_i$ and $A_i\neq B_i$ together imply $A_i\prec_i B_i$. Since $\preq_i$ is lexicographic, this further implies that $B_i$ contains the first good (according to the ordering of $\preq_i$ of the goods) that belongs to only one of the two sets $A_i,B_i$. Let $y_i = \argmax_{\preq_i} \{B_i\setminus A_i\}$ be that good, so that $A_i \setminus B_i \prec_i \{y_i\}$. As $\preq_i$ is a lexicographic preference we conclude that:
\begin{equation} \label{eq:1}
A_i =(A_i \cap B_i) \cup (A_i \setminus B_i) \prec_i (A_i \cap B_i) \cup \{y_i\}
\end{equation}

Define a directed graph $G=(V,E)$ as follows. For every agent $i\in V$, let $j\in \NN$ be the agent such that $y_i \in A_j$. Note that there must be such an agent because we assumed that $f$ must always allocate all the goods. Then in the transition from $A$ to $B$, agent $j$ transfers a good to agent $i$. But this implies $A_j\neq B_j$, therefore $j \in V$. So $j$ is also a vertex in $G$, and we can add the edge $(j,i)$ to the graph. Formally:
\[E = \{(j,i)\in V^2 : y_i\in A_j\}\]

Note that every vertex in $G$ has an in-degree of exactly 1. This implies the existence of a simple cycle $C=(c_1,\ldots, c_k)$ over $G$, where $k>1$ since there are no self-loops in $G$. For every index $j\in[k]$ let $x_{c_{j-1}}=y_{c_{j}}$ (using cyclic notation $c_{k}=c_0$) be the best good given to agent $c_{j}$ from $c_{j-1}$ in the transition from $A$ to $B$. Consider $\vecp{x}=(x_{c_1},\ldots, x_{c_k})$. Clearly $\vecp{x}$ is a single-good trading cycle, with augmented allocation $A^{\vecp{x}}_{i}=\left(A_{i}\setminus \{x_{i}\}\right)\cup \{y_{i}\}$ for every agent $i\in C$. To complete the proof, we will show it is strictly improving. For every index $j\in [k]$ we know that agent $c_j$ gives the good $x_{c_j}$ to the next agent in the transition from $A$ to $B$, so letting $i=c_j$ we get $x_{i}\in A_{i}\setminus B_{i}$, meaning that $A_{i}\cap B_{i} \subseteq A_{i}\setminus \{x_{i}\}$, and therefore $A_{i}\cap B_{i} \preq_{i} A_{i}\setminus \{x_{i}\}$. Combining this with the result from \cref{eq:1} we get:
\[A_{i}\prec_{i} \left(A_{i}\cap B_{i}\right) \cup \{y_{i}\}\preq_{i}\left(A_{i}\setminus \{x_{i}\}\right)\cup \{y_{i}\}=A^{\vecp{x}}_{i}\]
We have shown $A_{c_j}\prec_{c_j} A^{\vecp{x}}_{c_j}$ for every index $j\in [k]$, proving that $\vecp{x}$ is strictly improving.
\end{proof}

\subsection{\texorpdfstring{Proof of \cref{proposition:lexi-case}: High-Level Overview}
                        {Proof of Proposition~\ref{proposition:lexi-case}: High-Level Overview}}
\label{appendix:highlevel}

We wish to adapt the proofs of Hatfield \cite{Hatfield} for our model. The main difference between the two models is that Hatfield defines an ``allocation'' as one that always gives each agent exactly $Q$ goods, while we allow for any assignments of items to agents. Thus, when Hatfield discusses ``allocation mechanisms" he refers to mechanisms which are a-priori assumed to always assign every agent exactly $Q$ goods. This is an assumption that Hatfield uses in many stages of his proofs.

There are also differences in our assumptions: Hatfield proves his claims for agents with responsive preferences, while in this section we have the stronger assumption that agents preferences are lexicographic. Unlike us, Hatfield does not make use of neutrality when proving his claims, and only adds it in the final stage when proving his theorem.
On the other hand, Hatfield makes use of the additional assumption that the mechanism is Pareto-efficient {(compared to all allocations that assign exactly $Q$ goods to each agent)}. We, on the other hand, prove in (\cref{proposition:Pareto-efficient}) that Pareto-efficiency 
is implied by the other properties, so we do not need to assume it.

{The proof uses the concept of a ``consecutive allocation'' for lexicographic preferences:}
\begin{definition}
    An allocation $A=(A_1,\ldots, A_n)\in \ALLOCsM$ is \emph{consecutive} according to preference $\preq$ if there exists an ordering $p_1,\ldots , p_n$ of the agents s.t. for every pair of indices $i<j$ and goods $x\in A_{p_i},y\in A_{p_j}$ it holds that $y \preq x$.
\end{definition}
Intuitively, this means that if we order the goods from highest to lowest value according to the lexicographic preference $\preq$, under some ordering of the agents, the partition $A$ splits this interval into non-overlapping ``consecutive'' intervals.

The following lemma is based on Lemma 2 of  Hatfield \cite{Hatfield}. {We defer all proofs from this section to \cref{appendix:lex-hatfield}.}

\begin{restatable}{lemma}{idconsecutive}
\label{lemma:id-consecutive}
    Let $f$ be a partition mechanism for $n$ agents with preferences from the class $\Plex(\MM)$ of lexicographic preferences. If $f$ is truthful, non-bossy, and neutral then for any lexicographic preference $\preq\in \Plex$, for the identical profile $(\preq, \ldots ,\preq)$ the allocation $f(\preq, \ldots ,\preq)$ is consecutive according to $\preq$.
\end{restatable}

The following lemma proves a similar result to Lemma 3 of Hatfield \cite{Hatfield}, but because of the difference in the models the proofs are completely dissimilar.

\begin{restatable}{lemma}{globalquotaslexi}
\label{lemma:global-quotas-lexi}
    Let $f$ be a partition mechanism for $n$ agents with preferences from the class $\Plex(\MM)$ of lexicographic preferences.
     If $f$ is truthful, non-bossy, and neutral then there exists a quota-vector and ordering of the agents $\qp\in\Qnm$ such that for every lexicographic preference $\preq\in \Plex$, {for the identical profile $(\preq, \ldots ,\preq)$}  it holds that  $f(\preq,\ldots,\preq)=\SQqp(\preq,\ldots,\preq)$.
\end{restatable}

Following the arguments in the  proof of Theorem 1 of Hatfield, we extend the lemma and show that the same holds even when the lexicographic preferences are not necessarily identical, thus completing the proof for this section.

\lexicase*

\subsubsection{Missing Proofs} \label{appendix:lex-hatfield}

\idconsecutive*

\begin{proof}
Assume that $f$ is truthful, non-bossy, and neutral. Then \cref{proposition:Pareto-efficient} implies it is Pareto-efficient. Let $\preq\in \Plex$ be a lexicographic preference. We prove that the allocation $A$ defined $A_i=f_i(\preq, \ldots ,\preq)$ is consecutive according to $\preq$.

Assume for contradiction $A$ is not consecutive. Then w.l.o.g there exist distinct goods $c\preq b \preq a$ such that $a,c\in f_1(\preq, \ldots ,\preq)$ but $b\in f_2(\preq, \ldots ,\preq)$. Since $\preq$ is lexicographic and $a,b,c$ are distinct this implies $c\prec b \prec a$. For every agent $i\in \NN$ let $R_i =A_i\setminus \{a,b,c\}$ be the set of goods that agent $i$ gets without including $\{a,b,c\}$  and let $R_{-i}=\MM\setminus   (A_i\cup \{a,b,c\})$ be the set of goods that agent $i$ does not get, without including $\{a,b,c\}$. Note that for agents $i\in \NN\setminus\{1,2\}$, $R_i=A_i$. By definition:
\[f(\preq, \ldots ,\preq)=(A_1,\ldots, A_n)=(R_1\cup \{a,c\},R_2\cup \{b\}, R_3,\ldots, R_n)\]
For every $i\in \NN$ let $\vecp{e_i}$ be an ordering of the goods that is consistent with $(R_i, a,b,c,R_{-i})$ and let $\vecp{e'_i}$ be an ordering of the goods that is consistent with $(R_i, a,c,b,R_{-i})$, so that $\vecp{e},\vecp{e'}$ differ only in the ordering of $b,c$. Consider the lexicographic preferences profiles $\vecordp{e}$ and $\vecordp{e'}$. 

We will prove the following two claims:
\begin{enumerate}
    \item $\vecordp{e}$ is a push-up of $(\preq, \ldots ,\preq)$ for $f$.
    \item $\left(\preq_{\vecp{e_1}},\preq_{\vecp{e_2}},\preq_{\vecp{e'_{-1,2}}}\right)$ is a push-up of $\vecordp{e}$ for $f$.
\end{enumerate}
We will first show that with both of these claims combined we can prove the claim. Since $f$ is a truthful and non-bossy mechanism, \cref{monotonicity} combined with the above claims implies:
\[f\left(\preq_{\vecp{e_1}},\preq_{\vecp{e_2}},\preq_{\vecp{e'_{-1,2}}}\right)=f\vecordp{e}=f(\preq, \ldots ,\preq)=A\]
Let $\pi\in\perM$ that swaps $b,c$ while leaving the rest of the goods the same, we have $\preq_{\vecp{e'_i}}=\preq^{\pi}_{\vecp{e_i}}$. Then $\pi(A_i)=A_i$ for agents $i\in \NN\setminus\{1,2\}$ and $\pi(A_1)=R_1\cup\{a,b\}, \pi(A_2)=R_2\cup\{c\}$. Neutrality of $f$ implies:
\[f\vecordp{e'}=\pi\left(f\vecordp{e}\right)=\pi(A)=(R_1\cup \{a,b\},R_2\cup \{c\}, R_3,\ldots, R_n)\] 

Let $\preqa_1\in \Plex(\MM)$ a lexicographic preference corresponding to an ordering of the goods that is consistent with $(R_1,b,a,c,R_{-1})$. Consider the profile $\left(\preqa_1,\preq_{\vecp{e'_{-1}}}\right)$. Because $f_1\vecordp{e'}= R_1\cup\{a,b\}$, and since agent $1$ with preference $\preqa_1$ lexicographically places the goods of $R_1\cup\{a,b\}$ the highest, truthfulness of $f$ implies that $R_1\cup\{a,b\}\subseteq f_1\left(\preqa_1,\preq_{\vecp{e'_{-1}}}\right)$. Then $R_1\cup\{a,b\}\preq_{\vecp{e'_1}} f_1\left(\preqa_1,\preq_{\vecp{e'_{-1}}}\right)$. Yet truthfulness in the opposite direction implies that $f_1\left(\preqa_1,\preq_{\vecp{e'_{-1}}}\right)\preq_{\vecp{e'_1}} R_1\cup\{a,b\}$, so the strictness of lexicographic preference $\preq_{\vecp{e'_1}}$ implies $ f_1\left(\preqa_1,\preq_{\vecp{e'_{-1}}}\right)=R_1\cup\{a,b\}=f_1\vecordp{e'}$, and so non-bossiness of $f$ implies that $ f\left(\preqa_1,\preq_{\vecp{e'_{-1}}}\right)=f\vecordp{e'}=\pi(A)$. In particular, this implies that $f_2\left(\preqa_1,\preq_{\vecp{e'_{-1}}}\right)=\pi(A_2)=R_2\cup \{c\}$.

Finally, consider the profile $\left(\preqa_1,\preq_{\vecp{e_2}},\preq_{\vecp{e'_{-1,2}}}\right)$. Truthfulness in both directions with the profile $\left(\preqa_1,\preq_{\vecp{e'_{-1}}}\right)$, where agent $2$ gets $R_2\cup\{b\}$, implies that both $R_2\subseteq f_2\left(\preqa_1,\preq_{\vecp{e_2}},\preq_{\vecp{e'_{-1,2}}}\right)$ and $a\notin f_2\left(\preqa_1,\preq_{\vecp{e_2}},\preq_{\vecp{e'_{-1,2}}}\right)$. Assume for contradiction that $b\in f_2\left(\preqa_1,\preq_{\vecp{e_2}},\preq_{\vecp{e'_{-1,2}}}\right)$. Then $a\in f_1\left(\preqa_1,\preq_{\vecp{e_2}},\preq_{\vecp{e'_{-1,2}}}\right)$, otherwise agent $1$ would lie and report $\preq_{\vecp{e_1}}$, thus getting a better allocation $f_1\left(\preq_{\vecp{e_1}},\preq_{\vecp{e_2}},\preq_{\vecp{e'_{-1,2}}}\right)=A_1=R_1\cup \{a,c\}$ (as shown in Claim 1). But then Pareto-efficiency of $f$ is violated, contradicting \cref{proposition:Pareto-efficient}, because over profile $\left(\preqa_1,\preq_{\vecp{e_2}},\preq_{\vecp{e'_{-1,2}}}\right)$ agent $1$ with preference $\preqa_1$ would prefer to get $b$ over $a$ and agent $2$  with preference $\preqb_2$ would prefer $a$ over $b$.

Therefore $b\notin f_2\left(\preqa_1,\preq_{\vecp{e_2}},\preq_{\vecp{e'_{-1,2}}}\right)$. Then truthfulness with the profile $\left(\preqa_1,\preq_{\vecp{e'_{-1}}}\right)$ implies $c\in f_2\left(\preqa_1,\preq_{\vecp{e_2}},\preq_{\vecp{e'_{-1,2}}}\right)$, since preference $\preq_{\vecp{e_2}}$ lexicographically prefers $c$  over the entirety of $R_{-2}$. Then we have shown $f_2\left(\preqa_1,\preq_{\vecp{e'_{-1}}}\right)=R_2\cup \{c\}\subseteq f_2\left(\preqa_1,\preq_{\vecp{e_2}},\preq_{\vecp{e'_{-1,2}}}\right)$. Truthfulness in the opposite direction implies the subset cannot be strict, so $f_2\left(\preqa_1,\preq_{\vecp{e_2}},\preq_{\vecp{e'_{-1,2}}}\right)=f_2\left(\preqa_1,\preq_{\vecp{e'_{-1}}}\right)$. Then, non-bossiness implies $f\left(\preqa_1,\preq_{\vecp{e_2}},\preq_{\vecp{e'_{-1,2}}}\right)=f\left(\preqa_1,\preq_{\vecp{e'_{-1}}}\right)$. In particular, $f_1\left(\preqa_1,\preq_{\vecp{e_2}},\preq_{\vecp{e'_{-1,2}}}\right)=R_1\cup \{a,b\}$.

But now consider the profile $\left(\preq_{\vecp{e_1}},\preq_{\vecp{e_2}},\preq_{\vecp{e'_{-1,2}}}\right)$. Recall that we have shown in Claim 1 that $f_1\left(\preq_{\vecp{e_1}},\preq_{\vecp{e_2}},\preq_{\vecp{e'_{-1,2}}}\right)=A_1=R_1\cup \{a,c\}$. But then over this profile agent $1$ with preference $\preq_{\vecp{e_1}}$ would benefit from lying and reporting $\preqa_1$, thus getting allocation $R_1\cup\{a,b\}\succ_{\vecp{e_1}}R_1\cup \{a,c\}$, contradicting the truthfulness of $f$.

 \emph{Claim 1:} We prove that $\vecordp{e}$ is a push-up of $(\preq, \ldots ,\preq)$ for $f$. Let $i\in \NN$ and $Z\subseteq \MM$ such that $Z\preq A_i$, we will prove that $Z\preq_{\vecp{e_i}} A_i$.

First consider that case $i\notin \{1,2\}$, in this case $R_i=A_i$. Since $\preq_{\vecp{e_i}}$ places the goods of $R_i=A_i$ the highest, if it does not hold that $A_i\subseteq Z_i$ then $Z\preq_{\vecp{e_i}}A_i$ immediately. Otherwise, assume $A_i\subseteq Z$. Then $A_i\preq Z$, which combined with the assumption $Z\preq A_i$ and the strictness of lexicographic preferences implies $Z=A_i$, thus again proving $Z\preq_{\vecp{e_i}} A_i$.

Next consider the case $i=1$. In this case $A_1=R_1\cup \{a,c\}$. Since $\preq_{\vecp{e_1}}$ places the goods of $R_1\cup \{a\}$ the highest, if it does not hold that $R_1\cup \{a\}\subseteq Z$ then $Z\preq_{\vecp{e_1}}A_1$ immediately. Otherwise, assume $R_1\cup \{a\}\subseteq Z$. Since $c\prec b$, if it were the case that $b\in Z$ then $R_1\cup \{a,b\}\subseteq Z$, therefore:
\[Z\preq A_1=R_1\cup \{a,c\}\prec R_1\cup \{a,b\}\preq Z\]
A contradiction. Therefore $b\notin Z$. If $c\in Z$ then $A_1\subseteq Z$, implying $A_1\preq Z$, which combined with the assumption $Z\preq A_1$ implies $A_1=Z$, again proving $Z\preq_{\vecp{e_1}}A_1$. Otherwise, assume $c\notin Z$. Then $Z\subseteq R_1\cup \{a\} \cup R_{-1}$, since $\preq_{\vecp{e_1}}$ lexicographically prefers $c$ over the entirety of $R_{-1}$:
\[Z\preq_{\vecp{e_1}}  R_1\cup \{a\} \cup R_{-1}\preqa_1 R_1\cup \{a,c\} =A_1\]
Proving $Z\preq_{\vecp{e_1}}A_1$ as required.

Finally consider the case $i=2$. In this case $A_2=R_2\cup \{b\}$. Since $\preq_{\vecp{e_2}}$ places the goods of $R_2$ the highest, if it does not hold that $R_2\subseteq Z$ then $Z\preq_{\vecp{e_2}}A_2$ immediately. Otherwise, assume $R_2\subseteq Z$. Since $b\prec a$, if it were the case that $a\in Z$ then $R_1\cup \{a,b\}\subseteq Z$, therefore:
\[Z\preq A_2=R_2\cup \{b\}\prec R_2\cup \{a\}\preq Z\]
A contradiction. Therefore $a\notin Z$. If $b\in Z$ then $A_2\subseteq Z$, implying $A_2\preq Z$, which combined with the assumption $Z\preq A_2$ implies $A_2Z$, again proving $Z\preq_{\vecp{e_2}}A_2$. Otherwise, assume $b\notin Z$. Then $Z\subseteq R_1\cup \{a,c\} \cup R_{-1}$, since $\preq_{\vecp{e_2}}$ lexicographically prefers $b$ over the entirety of $\{c\}\cup R_{-2}$:
\[Z\preq_{\vecp{e_2}}  R_2\cup \{c\} \cup R_{-2}\preqa_1 R_2\cup \{b\} =A_2\]
Proving $Z\preq_{\vecp{e_2}}A_2$ as required.

\emph{Claim 2:} We make use of the result form Claim 1 that $\vecordp{e}$ is a push-up of $(\preq, \ldots ,\preq)$ for $f$. Since $f$ is truthful and non-bossy, this implies that $f\vecordp{e}=f(\preq, \ldots ,\preq)=A$. We prove that the profile $\left(\preq_{\vecp{e_1}},\preq_{\vecp{e_2}},\preq_{\vecp{e'_{-1,2}}}\right)$ is a push-up of $\vecordp{e}$ for $f$.

Let $i\in \NN\setminus\{1,2\}$ and $Z\subseteq \MM$ such that $Z\preq_{\vecp{e_i}}f_i\vecordp{e}=f_i(\preq, \ldots ,\preq)=A_i$, we prove that $Z\preq_{\vecp{e'_i}} A_i$. Note that since $i\notin \{1,2\}$, $R_i=A_i$ and $\preq_{\vecp{e'_i}}$ lexicographically places the goods of $R_i=A_i$ the highest. Therefore if it does not hold that $A_i\subseteq Z$ definitely $Z\preq_{\vecp{e'_i}} A_i$. Otherwise $A_i\subseteq Z$, then $A_i\preq_{\vecp{e_i}} Z$, which combined with the assumption $Z\preq_{\vecp{e_i}}A_i$ implies $A_i=Z$, thus again proving $Z\preq_{\vecp{e'
_i}} A_i$.
\end{proof}

The following lemma proves a similar result to Lemma 3 of Hatfield \cite{Hatfield}, but because of the difference in the models the proofs are completely dissimilar.

\globalquotaslexi*

\begin{proof}Fix a lexicographic preference $\preq\in \Plex$. \cref{lemma:id-consecutive} implies the partition $A=f(\preq, \ldots ,\preq)$ is consecutive according to the ordering of the goods defined by $\preq$. Meaning, there exists an ordering $\sigma\in \per{\NN}$ such that, letting $p_i=\sigma^{-1}(i)$, for every pair of indices $i<j$ and goods $x\in A_{p_i},y\in A_{p_j}$, $y\preq x$. For every index $i\in \NN$ let $q_i=|A_{p_i}|$ be the quota corresponding to the size of $A_{p_i}$.

By definition, for every index $i\in \NN$, every goods $x\in A_{p_i}$ is places higher by $\preq$ than all the goods in $\MM\setminus \bigcup_{j=1}^{i-1}A_{p_j}$. Then if we were to let an agent with lexicographic preference $\preq$ select $q_i$ goods from the set $\MM\setminus \bigcup_{j=1}^{i-1}A_{p_j}$ they would always take exactly the goods of $A_{p_i}$. Thus $A_{p_i}=\dem{\preq}{q_i}{\MM\setminus\bigcup_{j=1}^{i-1}A_{p_j}}$ for every index $i\in \NN$, proving that $f$  behaves like a $\qp$-serial-quota mechanism over the preference profile $(\preq, \ldots ,\preq)$. Now it is left to prove that for the same $\sigma, \vecp{q}$, the mechanism $f$ behaves like a $\qp$-serial-quota mechanism for all identical preference profiles.

Let $\preqa\in \Plex(\MM)$ be some lexicographic preference Let $\pi\in\perM$ be the permutation defined s.t. $\preqa=\preq^{\pi}$. Such a permutation must exist because both $\preqa,\preq^{\pi}$ are lexicographic and thus differ only in their ordering of the goods\footnote{{See \cref{appendix:subsec:lexi} for the formal proof.}}. By definition of demand, for every set $S\subseteq \MM$ and constant $k\in [m]$:
\[\dem{\preqa}{k}{\pi(S)}=\dem{\preq^{\pi}}{k}{\pi(S)}=\argmax_{\preq^{\pi}}\{T\subseteq \pi(S),\space |T|=k\}=\pi\left(\argmax_{\preq}\{T\subseteq S,\space |T|=k\}\right)=\pi\left(\dem{\preq}{k}{S}\right)\]
Meaning, the demand of $\preqa$ from the ``renamed" set $\pi(S)$ is just a permutation of the demand of $\preq$ from $S$. For every index $i\in \NN$, neutrality of $f$ implies:
\[f_{p_i}(\preqa, \ldots ,\preqa)=f_{p_i}(\preq^{\pi},\ldots, \preq^{\pi})=\pi(f_{p_i}(\preq, \ldots ,\preq))=\pi(A_{p_i})\]
Using the relationship from earlier:
\[\pi(A_{p_i})=\pi\left(\dem{\preq}{q_i}{\MM\setminus\bigcup_{j=1}^{i-1}A_{p_j}}\right)=\dem{\preqa}{q_i}{\pi\left(\MM\setminus\bigcup_{j=1}^{i-1}A_{p_j}\right)}=\dem{\preqa}{q_i}{\MM\setminus\bigcup_{j=1}^{i-1}\pi\left(A_{p_j}\right)}\]
So in total we have shown:
\[f_{p_i}(\preqa, \ldots ,\preqa)=\dem{\preqa}{q_i}{\MM\setminus\bigcup_{j=1}^{i-1}f_{p_j}(\preqa, \ldots ,\preqa)}\]
As required.
\end{proof}

\lexicase*
\begin{proof}It is trivial to show that any serial-quota mechanism is truthful, non-bossy, and neutral for strict preferences. We will prove the converse.

Let $\qp$ be as defined in \cref{lemma:global-quotas-lexi}. For each index $i\in[n]$ let $p_i=\sigma^{-1}(i)$ be the agent who chooses $i$-th as defined by $\sigma$. Let $\veco\in \Plex^n(\MM)$ be some lexicographic preference profile. 
We construct a new preference $\preqa\in \Plex(\MM)$ according to the following ordering of the goods. The first $q_1$ best goods will be $A_1=\dem{\preq_{p_1}}{q_1}{\MM}$, ordered arbitrarily. The next $q_2$ best goods will be $A_2=\dem{\preq_{p_2}}{q_2}{\MM\setminus A_1}$. In general, for $i\in[n]$, the goods ordered from index $(\sum_{j=1}^{i-1}q_j)+1$ to  index $\sum_{j=1}^{i}q_j$ according to $\preqa$ will be defined recursively:
\[A_i=\dem{\preq_{p_i}}{q_i}{\MM\setminus \bigcup_{j=1}^{i-1}A_j}\]
Note that by this definition $(A_1,\ldots, A_n)$ is exactly the allocation a $\qp$-serial-quota mechanism would output for the input $\veco$, with each agent in turn receiving their demand according to the ordering $\sigma$ and the quotas $\vecp{q}$. Then all that is left is to prove is that $f\veco=(A_1,\ldots , A_n)$, and this will complete the proof.

Consider the identical preference profile $(\preqa,\ldots, \preqa)$. \cref{lemma:id-consecutive} implies that $f(\preqa,\ldots, \preqa)=(A_1,\ldots, A_n)$. We will now show that $\veco$ is a push-up of $(\preqa,\ldots, \preqa)$ for $f$. As $f$ is a truthful and non-bossy allocation mechanism, by \cref{monotonicity} (monotonicity) this will imply that that $f\veco=f(\preqa,\ldots, \preqa)=A$, thus completing the proof. 

Let $i\in \NN$ and $Z\subseteq \MM$ such that $Z\preqa f_{p_i}(\preqa,\ldots, \preqa)=A_i$. We will show that $Z\preq_{p_i}A_i$. Recall from above:
\[f_{p_i}(\preqa,\ldots, \preqa)=A_{i}=\dem{\preq_{p_i}}{q_i}{\MM\setminus \bigcup_{j=1}^{i-1}A_j}\]
By the definition of $\preqa$, the set goods in $ \bigcup_{j=1}^{i-1}A_j$ consists of \textbf{exactly} all the goods ranked above the goods of $A_i$. Then since $\preqa$ is lexicographic $Z\preqa f_{p_i}(\preqa,\ldots, \preqa)$ implies $Z\cap \left(\bigcup_{j=1}^{i-1}A_j\right)=\emptyset$. If $A_i\subseteq Z$ then $A_i\preqa Z$, which combined with the assumption $Z\preqa A_i$ implies $Z=A_i$ and so $Z\preq_{p_i}A_i$. Otherwise, assume $A_i\nsubseteq Z$. Then there exists $x\in A_i$ such that $x\notin Z$. 

Assume for contradiction $A_i \prec_{p_i} Z$. Then there must exist some $y\in Z\setminus A_i$ such that $x \prec_{p_i} y$. Since we showed earlier that $Z\cap \left(\bigcup_{j=1}^{i-1}A_j\right)=\emptyset$ we know that $y\notin \bigcup_{j=1}^{i-1}A_j$. Yet now $x\in A_i=\dem{\preq_{p_i}}{q_i}{\MM\setminus \bigcup_{j=1}^{i-1}A_j}$ but $y\notin A_i$, so it must be that $y\prec_{p_i}x$, otherwise $y$ would be chosen for $A_i$ before $x$. We have reached a contradiction, therefore $Z\preq_{p_i} A_i$ as required.
\end{proof}

\section{\texorpdfstring{Missing Proofs from \cref{sec:strict-partition}}
                        {Missing Proofs from Section~\ref{sec:strict-partition}}}

\subsection{\texorpdfstring{Missing Proofs from \cref{subsec:control-lemma}}
                        {Missing Proofs from Section~\ref{subsec:control-lemma}}}
\label{appendix:control}

In this section we prove \cref{control-c1} and \cref{control-c2}, two assisting claims in the proof of \cref{control} (Control Lemma).

\begin{claim}
\label{control-c1}
    {In the context of the proof of \cref{control}, the preference profile $\vecordp{e}$ is a push-up of $\veco$ for $f$.}
\end{claim}
\begin{proof}
    Let $j\in \NN$ and $Z\subseteq \MM$ such that $Z \preq_j f_j\veco=A_j$, we will prove that $Z \preq_{\vecp{e}_j} A_j$. 

First consider the case $j=i$, recall that $A_i=R_i\cup S$. In this case agent $i$ with lexicographic preference $\preq_{\vecp{e}_i}$ strongly desires the set $R_i\cup S=A_i$. Then if it does not hold that $A_i\subseteq Z$ then clearly $Z\preq_{\vecp{e}_i}A_i$. Otherwise assume $A_i\subseteq Z$. In this case, $A_i\preq_i Z$, which combined with $Z\preq_i A_i$ and the assumption that $\PP$ is strict, implies that $A_i=Z$, thus again proving $Z\preq_{\vecp{e}_i}A_i$.

Next consider the case $j\neq i$, recall that $A_j=R_j$. In this case, since we assumed that  agent $j$ with preference $\preq_j$ strongly desires $S$ and $A_j\cap S=\emptyset$, the assumption $Z\preq_j A_j$ implies $Z\cap S =\emptyset$. Then since $\preq_{\vecp{e}_j}$ places the goods of $A_j=R_j$ highest after the goods of $S$, if it does not hold that $A_j\subseteq Z$ then immediately $Z\preq_{\vecp{e}_j}A_j$. Otherwise assume $A_j\subseteq Z$, then $A_j\preq_jZ$ which combined with $Z\preq_j A_j$ and the assumption that $\PP$ is strict implies $A_j=Z$, thus again proving $Z\preq_{\vecp{e}_j}A_j$. 
\end{proof}

\begin{claim}
\label{control-c2}
    {In the context of the proof of \cref{control}, the preference profile $\vecordp{e'}$ is a push-up of $\vecoa$ for $f$.}
\end{claim}
\begin{proof}
Let $j\in \NN$ and $Z\subseteq \MM$ such that $Z \preqa_j f_j\vecoa=A'_j$, we will prove that $Z \preq_{\vecp{e'_j}} A'_j$. 

First consider the case when $j\in \NN\setminus \{i\}$. Since preference $\preq_{\vecp{e'_j}}$ lexicographically places the goods of $A'_j$ the highest, if it does not hold that $A'_j\subseteq Z$ then immediately $Z\preq_{\vecp{e'_j}}A'_j$, as required. Otherwise $A'_j\subseteq Z$, implying $A'_j\preqa_jZ$ which combined with the assumption $Z\preqa_jA'_j$ and the strictness of $\PP$ implies $Z=A'_j$ and thus again $Z\preq_{\vecp{e'_j}}A'_j$, as required.

Next consider the case $j=i$. In this case $A'_i=R'_i\cup S_i$ Since preference $\preq_{\vecp{e'_i}}$ places the goods of $S_i$ the highest, if it does not hold that $S_i\subseteq Z$ immediately $Z \preq_{\vecp{e'_i}} S_i\preq_{\vecp{e'_i}}A'_i$ as required. Otherwise, assume $S_i\subseteq Z$. Since we originally assumed that agent $i$ with preference $\preqa_i$ strongly desires $S$, if $S_{-i}\cap Z\neq \emptyset$ then by definition of strong desire $A'_i=R'_i\cup S_{i} \preca_i Z$, contradicting the earlier assumption that $Z\preqa_iA'_i$, therefore this is impossible and so $S_{-i}\cap Z= \emptyset$.

Up to now we have $Z\cap S=S_i=A'_i\cap S$, it remains to see what goods $Z$ contains from $R_i\cup R_{-i}$. Since preference $ \preq_{\vecp{e'_i}}$ lexicographically prefers the goods from $R'_i$ over the entirety of $R'_{-i}$, then if it does not hold that $R'_i\subseteq Z$ this implies  $Z \preq_{\vecp{e'_i}} S_i\cup R'_i=A'_i$, as required. Otherwise, assume $R'_i\subseteq Z$, then $A'_i=R'_i\cup S_i\subseteq Z$ so $A'_j\preqa_jZ$, which combined with $Z\preqa_jA'_j$ implies $Z=A'_j$ and therefore again $Z \preq_{\vecp{e'_i}}A'_i$, as required.
\end{proof}

\subsection{\texorpdfstring{Necessity of all properties for the Control Lemma: {Proof of \cref{lemma:control-tight}}}
                        {Necessity of all properties for the Control Lemma: {Proof of Claim~\ref{lemma:control-tight}}}}
\label{appendix:control-nec}
{In this section we show that} {for every preference class $\PP\in \FFn(\MM)$, all three properties of truthfulness, non-bossiness, and neutrality are needed to prove the Control Lemma.}
\controltight*

\begin{proof}
{Let $\PP\in\FFn(\MM)$ be some permutations-closed class of strict preferences that contains the lexicographic preferences. We present three allocations mechanisms, each violates only one of the three properties (of truthfulness, non-bossiness, and neutrality), s.t. all three mechanisms do not satisfy the control claim. Since the only preferences that $\PP$ is guaranteed to contain are the lexicographic preferences, we will only use lexicographic preferences in our counterexamples.}
    
    \emph{Without truthfulness:} Consider the following non-bossy and neutral partition mechanism of     a set of goods $\MM=\{x,y\}$ among $n\geq 2$ agents.     If both agents have the same preference, then agent $1$ gets all the goods. Otherwise, agent $2$ gets all the goods. The mechanism is non-bossy as any partition mechanism for 2 agents is non-bossy. The mechanism is neutral since the property ``the preferences are identical'' is invariant under permutation of the names of the goods. Yet the mechanism is not truthful because agent $2$ benefits from lying when both agents have identical preferences. Finally, notice that the control claim does not hold{, even for lexicographic preferences}: over the profile where both agents have the same {{lexicographic}} preference which places the good $x$ highest, both agents strongly desire the singleton set $\{x\}$ and agent $1$ gets $x$, yet she does not control $x$ because if agent $2$ alters her reported preference she can take $x$, so this contradicts the control claim.
    
\emph{Without non-bossiness:}  Consider the following truthful and neutral partition mechanism of a set of goods $\MM$ among $n\geq 3$ agents. Let $x$ be the highest rated good according to agent $1$, then agent $2$ is allocated $\{x\}$ while agent $3$ is allocated all remaining goods. The mechanism is truthful because the only agent that affects the allocation is agent $1$, but she never gets any goods anyway so she can never have incentive to lie. The mechanism is neutral since permuting the names of the goods causes the best good $x$ to be permuted accordingly. Yet the mechanism is bossy, since when agent $1$ swaps her highest rated good she affects the allocation of agent $2$ without changing her own allocation. Finally, notice that the control claim does not hold{, even for lexicographic preferences}: consider the case when all agents report {{lexicographic}} preferences {which place some good $x\in \MM$ the highest. Then all agents strongly desire the set} $\{x\}$, and agent $2$ gets $\{x\}$ whole. Yet if agent $1$ reports some other preference which {places a different good the highest} then agent $3$ gets $x$ instead, even though agent $2$ still strongly desires it, contradicting the control claim. We remark that one can also modify this example to show that even with $n=2$ agents, a similar claim holds for allocations (not partitions). This can be done by instead of giving goods to agent $3$, simply throwing them away.

    \emph{Without neutrality:}  Consider the following truthful and non-bossy partition mechanism for allocating {a set $\MM$ of $m\geq 3$ goods} 
        among $n\geq 3$ agents.     Fix two distinct goods $a\neq b\in \MM$. Agent $1$ first picks her single favorite good $x$ from $\MM\setminus \{a\}$ and receives it. There are now two cases:
    \begin{enumerate}
        \item If $x=b$ then agent $2$ picks her single favorite good $y$ from $\MM\setminus \{x\}$ and agent $3$ gets all remaining goods $\MM\setminus\{x,y\}$. 
        \item If $x\neq b$ then agent $3$ picks her single favorite good $y$ from $\MM\setminus \{x\}$ and agent $2$ gets all remaining goods $\MM\setminus\{x,y\}$.
    \end{enumerate}
The mechanism is truthful since lying can only cause an agent to {get} a less preferred good when choosing. The mechanism is also non-bossy since the agents can only alter the overall allocation by altering the good they select at some step, thus altering their own allocation as well. Yet the mechanism is not neutral since it makes use of the specific naming of two goods $a,b$, so if we rename the goods according to a permutation that swaps the names of one of these goods with a different good the resulting allocation may be radically different. Finally, notice that control claim does not hold{, even for lexicographic preferences}: consider two different profiles when all agents {have lexicographic preferences which place good $a$ the highest, and thus} strongly desire the singleton set $\{a\}$, s.t. in the first profile {agent $1$'s second favorite good is $b$, while in the second profile agent $1$'s second favorite good is some other good $c\in \MM\setminus\{a,b\}$.} In both profiles all agents strongly desire $\{a\}$, yet in the first agent $2$ gets $a$ while in the second profile agent $3$ gets $a$, contradicting the control claim. 
\end{proof}

\subsection{\texorpdfstring{Missing Proofs from \cref{subsec:strict-partition}}
                        {Missing Proofs from Section~\ref{subsec:strict-partition}}}
\label{appendix:strict-partition}

For the following lemmas let $f$ be a truthful, non-bossy, and neutral partition mechanism of a set $\MM$ of $m$ goods among a set $\NN$ of $n$ agents
with  preferences from some class $\PP\in \FFn(\MM)$. 
Moreover let $\qp\in\Qnm$ be the quota-vector and ordering of the agents for which $f$ is a serial-quota mechanism according to \cref{lexi-serial-quota}.

In this section, we prove the following characterization for partition mechanisms:
\strictpartition*

Given a preference $\preq$ over a set of goods $\MM$, for any subset of goods $\WW\subseteq\MM$ {we consider the preference $\restW{\preq}$ \emph{induced} by $\preq$ on $\WW$: the
one in which preference over subsets of $\WW$ are simply determined according to $\preq$.  Formally:}

\begin{definition}
\label{definition:induced}
    Let $\preq$ by a preference over a set $\MM$ of $m$ goods. For a subset $\WW\subseteq \MM$, the preference $\restW{\preq}$ denotes the preference \emph{induced} by $\preq$ on $\WW$, s.t. for every $S,T\subseteq W$:       \[S\restW{\preq}T\iff S\preq T\]
    Let $\PP$ be a preference class over $\MM$. The \emph{class of preferences induced} by $\PP$ on $\WW$ is the class $\PP_{\WW}=\{\restW{\preq}:\preq\in\PP\}$.     Note that if $\PP\in \FFn(\MM)$ then $\PP_{\WW}\in \FFn(\WW)$.
\end{definition}

{We first show that after agent $p_1$ takes her $q_1$-demand, 
the allocation of the remaining agents is decided only according to their preferences on the remaining goods.} {That is, the preferences of the agents on the goods taken by agent $p_1$ cannot affect the outcome of $f$. We formalize this using \cref{definition:induced} of \emph{induced preferences} {in the following lemma}.

\begin{restatable}{lemma}{firstpickindependence}
    \label{first-pick-independence}
    Let $\PP\in \FFn(\MM)$ be a class of \emph{strict} preferences. Fix a strict preference $\preq_{p_1}\in \PP$, let $D=\dem{\preq_{p_1}}{q_1}{\MM}$ be the unique $q_1$-demand of agent $p_1$ with preference $\preq_{p_1}$, and let $\WW=\MM\setminus D$ denote the set of remaining goods after agent $p_1$ has taken her $q_1$-demand. For any pair of profiles $\preq_{-p_1}, \preqa_{-p_1}\in \PP^{n-1}$ such that $\restW{\preq_i}=\restW{\preqa_i}$ for every $i\in \NN\setminus\{p_1\}$:    \[f(\preq_{p_1},\preq_{-p_1})=f(\preq_{p_1},\preqa_{-p_1})\]
\end{restatable}

\begin{proof}
    Assume w.l.o.g that $p_1=1$. We will prove the claim by induction on $i$, from $1$ to $n$, showing that:
    \[f(\preq_1, \preqa_2,\ldots , \preqa_{i},\preq_{i+1},\ldots , \preq_n)=f\veco\]
    
    \textit{Base Case:} $i=1$. This case is trivial since the LHS preference profile is just $\veco$.
    
    \textit{Induction Case:} Let $i>1$, assume that $f(\preq_1, \preqa_2,\ldots , \preqa_{i-1},\preq_{i},\ldots , \preq_n)=f\veco$, we will prove that $f(\preq_1, \preqa_2,\ldots , \preqa_{i},\preq_{i+1},\ldots , \preq_n)=f\veco$. \cref{first-picker} implies that over any profile where agent $p_1=1$ reports preference $\preq_1$ she gets allocation $D=\dem{\preq_{p_1}}{q_1}{\MM}$, and so the allocations of all other agents are subsets of $\WW$. In particular, this implies both:
    \[f_i(\preq_1, \preqa_2,\ldots , \preqa_{i-1},\preq_{i},\ldots , \preq_n)\subseteq \WW\]
    \[f_i(\preq_1, \preqa_2,\ldots , \preqa_{i},\preq_{i+1},\ldots , \preq_n)\subseteq \WW\]
    
    The profile $(\preq_1, \preqa_2,\ldots , \preqa_{i},\preq_{i+1},\ldots , \preq_n)$ is obtained from $(\preq_1, \preqa_2,\ldots , \preqa_{i-1},\preq_{i},\ldots , \preq_n)$ by having agent $i$ swap her preference from $\preq_i$ to $\preqa_{i}$. Assume for contradiction that:
    \[f_i(\preq_1, \preqa_2,\ldots , \preqa_{i-1},\preq_{i},\ldots , \preq_n) \neq f_i(\preq_1, \preqa_2,\ldots , \preqa_{i},\preq_{i+1},\ldots , \preq_n)\]
    
    Since both allocations are subsets of $\WW$ and $\restW{\preq_i}=\restW{\preqa_i}$, this implies that both preferences $\preq_i$ and $\preqa_i$ must agree which one is of the allocations is better (and since both preferences are strict, one of the allocations be \emph{strictly} better). But this contradicts the truthfulness of $f$, since over the profile with the worse allocation agent $i$ would benefit from lying and swapping her reported preference to the other one. Then $f_i(\preq_1, \preqa_2,\ldots , \preqa_{i-1},\preq_{i},\ldots , \preq_n) = f_i(\preq_1, \preqa_2,\ldots , \preqa_{i},\preq_{i+1},\ldots , \preq_n)$, so non-bossiness of $f$ implies: \[f(\preq_1, \preqa_2,\ldots , \preqa_{i-1},\preq_{i},\ldots , \preq_n) = f(\preq_1, \preqa_2,\ldots , \preqa_{i},\preq_{i+1},\ldots , \preq_n)=f\veco\]

Proving the induction case.
    
    Taking the end of the induction $i=n$ we get:    \[f(\preq_1, \preqa_2,\ldots , \preqa_{n})=f\veco\]
    
    Proving the claim.
\end{proof}

We have seen in \cref{first-pick-independence} that after agent $p_1$ takes her $q_1$-demand, the allocation of the remaining goods among the remaining agents depends only on their preferences over the remaining goods. This allows us to define a new partition mechanism $f^{\preq_{p_1}}$ which allocates the set $\WW$ of the remaining goods  among the remaining $n-1$ agents, such that the mechanism receives from the agents only their induced preferences on $\WW$, and is identical to $f$. 
{We refer to $f^{\preq_{p_1}}$ as \emph{induced partition mechanism}}.
In the following lemma, we prove that this mechanism is also  truthful, non-bossy, and neutral. The purpose of this lemma is to facilitate the use of an induction argument in \cref{strict-partition}, by showing how to ``remove" the first agent and her bundle from the mechanism.

\begin{restatable}{lemma}{inducedmechanism}
    \label{induced-mechanism}
     Let $f$ be a truthful, non-bossy, and neutral partition mechanism of a set $\MM$ goods among a set $\NN$ of $n$ agents with \emph{strict} preferences from some class $\PP\in \FFn(\MM)$. Fix a strict preference $\preq_{p_1}\in \PP$, let $D=\dem{\preq_{p_1}}{q_1}{\MM}$ be the unique $q_1$-demand of $p_1$ with preference $\preq_{p_1}$. Let $\WW=\MM\setminus D$ denote the set of remaining goods after agent $p_1$ has taken her $q_1$-demand. Then exists an {\emph{induced partition mechanism}}
     $f^{\preq_{p_1}}$ of the goods $\WW$ among $n-1$ agents $\NN'=\NN\setminus\{p_1\}$ with strict preferences from class $\PP_{\WW}\in \FFn(\WW)$ such that, letting $f_{-p_1}$ denote the allocation of $f$ to all agents except agent $p_1$, for every profile $\preq_{-p_1}\in \PP^{n-1}$:  \[f^{\preq_{p_1}}(\restW{\preq_{-p_1}})=f_{-p_1}\veco\]
     Moreover, $f^{\preq_{p_1}}$ is truthful, non-bossy, and neutral.
\end{restatable}

\begin{proof}
                Fix a preference $\preq_{p_1}\in \PP$. We must first prove that the partition mechanism $f^{\preq_{p_1}}$ is well-defined. That is, we will prove that for every input preference profile $\preq_{-p_1}\in \PP_{\WW}^{n-1}$ with preferences over $\WW$, the mechanism $f^{\preq_{p_1}}$ outputs a single partition. Note that this is not immediately true from the above definition since there can be many different preferences of the set of goods $\MM$ that reduce to the same preference when induced {on} $\WW$. Formally, let $\preq_{-p_1}, \preqa_{-p_1}\in \PP^{n-1}$ be profiles such that $\restW{\preq_i}=\restW{\preqa_i}$ for every agent $i\in \NN\setminus\{p_1\}$. Then \cref{first-pick-independence} implies that $f(\preq_{p_1},\preq_{-p_1})=f(\preq_{p_1},\preqa_{-p_1})$, so we can define $f^{\preq_{p_1}}(\restW{\preq_i})=f^{\preq_{p_1}}(\restW{\preqa_i})=f_{-p_1}(\preq_{p_1},\preq_{-p_1})$ and $f^{\preq_{p_1}}$ is well-defined. 

        We will now prove that $f^{\preq_{p_1}}$ is truthful, non-bossy, and neutral.

        \emph{Truthfulness:} Let $\preq_{-p_1}\in \PP_{\WW}^{n-1}$ be some preference profile of the agents $\NN\setminus\{p_1\}$ over the set of goods $\WW$, let $i\in\NN\setminus\{p_1\}$ be some agent an let $\preqa_i\in \PP_{\WW}$ be an alternate preference. In order to prove the truthfulness of $f^{\preq_{p_1}}$ we must prove that:        \[f^{\preq_{p_1}}_i(\preqa_i,\preq_{-i,p_1})\preq_if^{\preq_{p_1}}_i\left(\preq_i,\preq_{-i,p_1}\right)\]
Let $\preqb_{-p_1}\in \PP^{n-1}$ be a preference profile over the original set of goods $\MM$ that satisfies $\restW{\preqb_j}=\preq_j$  for every agent $j\in\NN\setminus\{p_1\}$. Such a profile must exist from the definition of  $\PP_{\WW}$ (see \cref{definition:induced}). These preferences order the subsets of $\WW$ exactly as their counterparts from $\preq_{-p_1}$, with the addition of also ordering the goods of $\MM$ in some arbitrary manner. Similarly, define a preference $\preqb'_i\in \PP$ that satisfies $\restW{\preqb'_i}=\preqa_i$. Then by definition of $f^{\preq_{p_1}}$:
\[f_{-p_1}\left(\preq_{p_1},\preqb_i,\preqb_{-i,p_1}\right)=f^{\preq_{p_1}}\left(\restW{\preqb_i},\restW{\preqb_{-i,p_1}}\right)=f^{\preq_{p_1}}\left(\preq_i,\preq_{-i,p_1}\right)\]
\[f_{-p_1}\left(\preq_{p_1},\preqb'_i,\preqb_{-i,p_1}\right)=f^{\preq_{p_1}}\left(\restW{\preqb'_i},\restW{\preqb_{-i,p_1}}\right)=f^{\preq_{p_1}}(\preqa_i,\preq_{-i,p_1})\]
Truthfulness of $f$ implies that:
\[f^{\preq_{p_1}}_i(\preqa_i,\preq_{-i,p_1})=f_i\left(\preq_{p_1},\preqb'_i,\preqb_{-i,p_1}\right)\preqb_i f_i\left(\preq_{p_1},\preqb_i,\preqb_{-i,p_1}\right)=f^{\preq_{p_1}}_i\left(\preq_i,\preq_{-i,p_1}\right)\]
Finally, since we defined $\preqb_i$ so that $\restW{\preqb_i}=\preq_i$ and both $f^{\preq_{p_1}}_i\left(\preq_i,\preq_{-i,p_1}\right),f^{\preq_{p_1}}_i(\preqa_i,\preq_{-i,p_1})\subseteq \WW$, the above inequality is equivalently stated:
\[f^{\preq_{p_1}}_i(\preqa_i,\preq_{-i,p_1})\preq_if^{\preq_{p_1}}_i\left(\preq_i,\preq_{-i,p_1}\right)\]
Proving the truthfulness of $f^{\preq_{p_1}}$.

\emph{Non-Bossiness:} Let $\preq_{-p_1}\in \PP_{\WW}^{n-1}$ be some preference profile of the agents $\NN\setminus\{p_1\}$ over the set of goods $\WW$, let $i\in\NN\setminus\{p_1\}$ be some agent an let $\preqa_i\in \PP_{\WW}$ be an alternate preference. Since the preference class $\PP_{\WW}$ contains only strict preferences, in order to prove that $f^{\preq_{p_1}}$ is neutral we assume that $f^{\preq_{p_1}}_i(\preqa_i,\preq_{-i,p_1})=f^{\preq_{p_1}}_i\left(\preq_i,\preq_{-i,p_1}\right)$ and prove that $f^{\preq_{p_1}}(\preqa_i,\preq_{-i,p_1})=f^{\preq_{p_1}}\left(\preq_i,\preq_{-i,p_1}\right)$. 

Let $\preqb_{-p_1}\in \PP^{n-1}$ be a preference profile over the original set of goods $\MM$ that satisfies $\restW{\preqb_j}=\preq_j$  for every agent $j\in\NN\setminus\{p_1\}$. Similarly, define a preference $\preqb'_i\in \PP$ that satisfies $\restW{\preqb'_i}=\preqa_i$. Then by definition of $f^{\preq_{p_1}}$:
\[f_{-p_1}\left(\preq_{p_1},\preqb_i,\preqb_{-i,p_1}\right)=f^{\preq_{p_1}}\left(\restW{\preqb_i},\restW{\preqb_{-i,p_1}}\right)=f^{\preq_{p_1}}\left(\preq_i,\preq_{-i,p_1}\right)\]
\[f_{-p_1}\left(\preq_{p_1},\preqb'_i,\preqb_{-i,p_1}\right)=f^{\preq_{p_1}}\left(\restW{\preqb'_i},\restW{\preqb_{-i,p_1}}\right)=f^{\preq_{p_1}}(\preqa_i,\preq_{-i,p_1})\]
Since we assumed that $f^{\preq_{p_1}}_i(\preqa_i,\preq_{-i,p_1})=f^{\preq_{p_1}}_i\left(\preq_i,\preq_{-i,p_1}\right)$ from the above equations this implies $f_i\left(\preq_{p_1},\preqb_i,\preqb_{-i,p_1}\right)=f_i\left(\preq_{p_1},\preqb'_i,\preqb_{-i,p_1}\right)$. Then non-bossiness of $f$ implies that $f\left(\preq_{p_1},\preqb_i,\preqb_{-i,p_1}\right)=f\left(\preq_{p_1},\preqb'_i,\preqb_{-i,p_1}\right)$. Then the above equations imply that:
\[f^{\preq_{p_1}}(\preqa_i,\preq_{-i,p_1})=f_{-p_1}\left(\preq_{p_1},\preqb'_i,\preqb_{-i,p_1}\right)=f^{\preq_{p_1}}\left(\preq_i,\preq_{-i,p_1}\right)=f^{\preq_{p_1}}\left(\preq_i,\preq_{-i,p_1}\right)\]
Proving the non-bossiness of $f^{\preq_{p_1}}$.

\emph{Neutrality:} Let $\preq_{-p_1}\in \PP_{\WW}^{n-1}$ be some preference profile and let $\pi\in\per{W}$ be a permutation of the set of goods $\WW$. Since the preference class $\PP_{\WW}$ contains only strict preferences, in order to prove that $f^{\preq_{p_1}}$ is neutral we must prove that $f^{\preq_{p_1}}(\preq_{-p_1}^{\pi})= \pi(f^{\preq_{p_1}}(\preq_{-p_1}))$. Let $\sigma\in \perM$ be a permutation of $\MM$ defined from $\pi$ as follows:
\[\sigma(x)=\begin{cases}
    \pi(x) & x\in \WW \\
    x & x\in \MM\setminus \WW
\end{cases}\]
That is, $\sigma$ permutes the goods of $\WW$ exactly like $\pi$ and behaves like the identity for all other goods.

Let $\preqb_{-p_1}\in \PP^{n-1}$ be a preference profile over the original set of goods $\MM$ that satisfies $\restW{\preqb_j}=\preq_j$  for every agent $j\in\NN\setminus\{p_1\}$. We will first prove that $\restW{\preqb^{\sigma}_j}=\preq^{\pi}_j$ for every $j\in\NN\setminus\{p_1\}$. Meaning, we show that if we take a preference over the set of goods $\MM$, permute it according to $\sigma$ and then induce it to $\WW\subseteq \MM$, this is equivalent as first inducing it to $\WW$ and then permuting according to $\pi$. Fix some $j\in\NN\setminus\{p_1\}$ and let $S,T\subseteq \WW$, we must show that $S\preqb^{\sigma}_j T\iff S\preq^{\pi}_j T$. By definition of the permutation of a preference, $S\preqb^{\sigma}_j T \iff \sigma^{-1}(S)\preqb_j \sigma^{-1}(T)$. From the definition of $\sigma$, since both $S,T$ are in $\WW$ we get $\sigma^{-1}(S)=\pi^{-1}(S)$ and $\sigma^{-1}(T)=\pi^{-1}(T)$, so $\sigma^{-1}(S)\preqb_j \sigma^{-1}(T)\iff \pi^{-1}(S)\preqb_j \pi^{-1}(T)$. Since we originally defined $\preqb_j$ so that $\restW{\preqb_j}=\preq_j$ and both $\pi^{-1}(S),\pi^{-1}(T)\subseteq \WW$ we get $\pi^{-1}(S)\preqb_j \pi^{-1}(T)\iff \pi^{-1}(S)\preq_j \pi^{-1}(T)$. Finally, using the definition of the permutation of a preference, we know that $\pi^{-1}(S)\preq_j \pi^{-1}(T)\iff S\preq^{\pi}_j T$, completing the proof that $\restW{\preqb^{\sigma}_j}=\preq^{\pi}_j$.

Consider the preference $\preq_{p_1}^{\sigma}$. Since we defined $\sigma$ such that $\sigma(x)=x$ for all goods $x\in \MM \setminus \WW  =D= \dem{\preq_{p_1}}{q_1}{\MM}$, the demand of agent $p_1$ will remain the same with preference $\preq_{p_1}^{\sigma}$, so $f_{p_1}(\preq^{\sigma}_{p_1},\preqb^{\sigma}_{-p_1})=f_{p_1}(\preq_{p_1},\preqb^{\sigma}_{-p_1})=D$. Therefore non-bossiness of $f$ implies that $f(\preq^{\sigma}_{p_1},\preqb^{\sigma}_{-p_1})=f(\preq_{p_1},\preqb^{\sigma}_{-p_1})$, and so neutrality of $f$ implies:
\[f(\preq_{p_1},\preqb^{\sigma}_{-p_1})=f(\preq^{\sigma}_{p_1},\preqb^{\sigma}_{-p_1})=\sigma(f(\preq_{p_1},\preqb_{-p_1}))\]
Since $f_{p_1}(\preq_{p_1},\preqb_{-p_1})=D=\MM\setminus \WW$, we know that $f_{j} (\preq_{p_1},\preqb_{-p_1})\subseteq \WW$ for every agent $j\in\NN\setminus\{p_1\}$, so $\sigma(f_{j} (\preq_{p_1},\preqb_{-p_1}))=\pi(f_{j} (\preq_{p_1},\preqb_{-p_1}))$. Finally, using the definition of $f^{\preq_{p_1}}$ and that both $\restW{\preqb_j}=\preq_j$ and $\restW{\preqb^{\sigma}_j}=\preq^{\pi}_j$ for every $j\in\NN\setminus\{p_1\}$, we get:
\[f^{\preq_{p_1}}(\preq^{\pi}_{-p_1})=f_{-p_1}(\preq_{p_1},\preqb^{\sigma}_{-p_1})=\sigma(f_{-p_1}(\preq_{p_1},\preqb_{-p_1}))=\pi(f_{-p_1}(\preq_{p_1},\preqb_{-p_1}))=\pi(f^{\preq_{p_1}}(\preq_{-p_1}))\]
Proving the neutrality of $f^{\preq_{p_1}}$.
\end{proof}

Finally, we prove the characterization.

\begin{proof} [Proof of \Cref{strict-partition}]
As preferences are strict, proving that a serial-quota mechanism is truthful, non-bossy, and neutral is trivial. We will now prove the converse. 

We will prove by induction on $n$, the number agents. The base step $n=1$ is trivial: when there is only one agent and all goods must be allocated, she must always get all the goods. Then, choosing $p_1=1,q_1=m$ we get that for any $\preq_1\in \PP$, $f_1(\preq_1)=\MM=\dem{\preq_1}{m}{\MM}$ (as her preference is strict and monotone, $\MM$ is indeed her demand), therefore $f$ is a $\qp$-serial-quota partition mechanism as required. We will now prove the induction step.

For the induction hypothesis, we can assume the lemma holds for all partition mechanisms of any set of goods $\WW$ to $n-1$ agents. By assumption, the partition mechanism $f$ is a truthful, non-bossy, and neutral mechanism of the set of goods $\MM$ among a set $\NN$ of $n$ agents with strict preferences from 
class $\PP\in \FFn(\MM)$. Let $\qp\in\Qnm$ be the quota-vector and ordering of the agents defined for $f$ in \cref{lexi-serial-quota}. We will prove that $f$ is a $\qp$-serial-quota partition mechanism.

Let $\veco\in \PP^n$ be some preference profile. Denote $D=\dem{\preq_{p_1}}{q_1}{\MM}$, then \cref{first-picker} implies that $f_{p_1}\veco =D$.

Let $\vecp{e}=(e_1,\ldots , e_m)$ be an ordering of the goods of $\MM$ that places the goods of $D$ first. Consider the identical lexicographic profile $(\preq_{\vecp{e}},\ldots , \preq_{\vecp{e}})$. \cref{lexi-serial-quota}  implies that $f(\preq_{\vecp{e}},\ldots , \preq_{\vecp{e}})=\SQqp(\preq_{\vecp{e}},\ldots , \preq_{\vecp{e}})$. Since all agents order the goods the same, each agent will take in turn a prefix of the remaining goods, according to their quota. Formally, letting $Q_i=\sum_{j=1}^iq_j$ denote the partial sums of the quotas, we get $f_{p_i}(\preq_{\vecp{e}},\ldots , \preq_{\vecp{e}})=\{e_{Q_{i-1}+1},\ldots , e_{Q_i}\}$ for every index $i\in[n]$ (using notation $Q_0=0$). Specifically note that $f_{p_1}(\preq_{\vecp{e}},\ldots , \preq_{\vecp{e}})=\{e_1,\ldots , e_{q_1}\}=D$.

\cref{induced-mechanism} shows that there is a truthful, non-bossy, and neutral partition mechanism $f^{\preq_{\vecp{e}}}$ that partitions the set $\WW=\MM\setminus D=\{e_{q_1+1},\ldots , e_m\}$ among $n-1$ agents $\NN'=\NN\setminus\{p_1\}$ with preferences from $\PP_{\WW}$ s.t. $f^{\preq_{\vecp{e}}}(\restW{\preqb_{-p_1}})=f_{-p_1}(\preq_{\vecp{e}},\preqb_{-p_1})$ for every profile $\preqb_{-p_1}\in \PP^{n-1}$. Let $\vecp{e}_{\WW}=(e_{q_1+1},\ldots , e_m)$ denote the ordering of the goods of $\WW$ according to $\vecp{e}$, so that $\restW{\preq_{\vecp{e}}}=\preq_{\vecp{e}_{\WW}}$. Consider the identical preference profile $(\preq_{\vecp{e}_{\WW}},\ldots , \preq_{\vecp{e}_{\WW}})\in \Plex^{n-1}(\WW)$. By definition of $f^{\preq_{\vecp{e}}}$ we get for every index $i\in \NN'$:
\[f^{\preq_{\vecp{e}}}_{p_i}(\preq_{\vecp{e}_{\WW}},\ldots , \preq_{\vecp{e}_{\WW}})=f_{p_i}(\preq_{\vecp{e}},\ldots , \preq_{\vecp{e}})=\{e_{s_{i-1}+1},\ldots , e_{Q_i}\}\]

On the other hand, the induction hypothesis implies that $f^{\preq_{\vecp{e}}}$ is a $\qpa$-serial-quota mechanism for some {\emph{unique}} $\qpa\in \Q{n-1}{|\WW|}$. Yet, the above equation implies that this is only possible if $p'_i=p_{i+1}$ and $q'_i=q_{i+1}$ for every index $i\in[n]$. That is, the $i$-th agent to select a good in $f^{\preq_{\vecp{e}}}$ is the same as the $i+1$-th agent in $f$, with the same quota.

 Note that for any preference $\preqb\in \PP$, quota $q$ and set of goods $S\subseteq \WW$, $\dem{\preqb}{q}{S}=\dem{\restW{\preqb}}{q}{S}$, meaning the demand of $\preqb$ from subsets of $\WW$ stays the same when considering the induced preference $\restW{\preqb}$ on $\WW$ by $\preqb$. Then for any profile $\preqb_{-p_1}\in \PP^{n-1}$ we know that $\SQqpa(\restW{\preqb_{-p_1}})=\SQqp_{-p_1}(\preq_{\vecp{e}},\preqb_{-p_1})$, since each agent will have the exact same demand at each stage.

Consider the profile $(\preq_{\vecp{e}},\preq_{-p_1})$. This profile is obtained from profile $\veco$ when agent $p_1$ swaps her preference from $\preq_{p_1}$ to $\preq_{\vecp{e}}$, both of which have the same $q_1$-demand from $\MM$, meaning $f_{p_1}\veco=D=f_{p_1}(\preq_{\vecp{e}},\preq_{-p_1})$. Then non-bossiness implies that $f\veco=f(\preq_{\vecp{e}},\preq_{-p_1})$. On the other hand, for every index $i>1$ we know that:
\[f_{p_i}(\preq_{\vecp{e}},\preq_{-p_1})=f^{\preq_{\vecp{e}}}_{p_i}(\restW{\preq_{-p_1}})=\SQqpa_{p_i}(\restW{\preq_{-p_1}})=\SQqp_{p_i}(\preq_{\vecp{e}},\preqb_{-p_1})\]
Thus:
\[f\veco =f(\preq_{\vecp{e}},\preq_{-p_1})=\SQqp(\preq_{\vecp{e}},\preqb_{-p_1})=\SQqp\veco\]
Completing the proof.

\end{proof}

\subsection{\texorpdfstring{Missing Proofs from \cref{sec:strict-alloc}}
                        {Missing Proofs from Section~\ref{sec:strict-alloc}}}
\label{appendix:alloc-strict}
{
\strict*
\begin{proof}
It is trivial to show that any serial-quota mechanism is truthful, non-bossy, and neutral for strict preferences. We will prove the converse.

Let $\qp\in\Qnm$ be as defined in \cref{strict-partition}. We will prove that for every profile $\veco\in \PP^n$, $f\veco=\SQqp\veco$. Let $h$ be the partition mechanism of $\MM$ among $n+1$ agents (with preferences from $\PP$) defined as follows. 
For every profile $(\preq_1,\ldots , \preq_n,\preq_{n+1})\in \PP^{n+1}$ and agent $i\in \NN$:
\[h_i(\preq_1,\ldots , \preq_n,\preq_{n+1})=f_i\veco\]
And for agent $n+1$:
\[h_{n+1}(\preq_1,\ldots , \preq_n,\preq_{n+1})=\MM\setminus \left(\bigcup_{i=1}^n f_i\veco \right)\]
Meaning $h$ allocates all the goods $f$ does not allocated to agent $n+1$, thus always partitioning $\MM$. It is easy to see that the truthfulness and non-bossiness of $f$ imply the same holds for $h$. We will prove that $h$ is also neutral. Let $(\preq_1,\ldots , \preq_n,\preq_{n+1})\in \PP^{n+1}$ be a preference profile and $\pi\in\perM$ be a permutation of the goods. Since we assumed $f$ is neutral, for all $i\in\NN$:
\[f_i\vecopi=\pi(f_i\veco)\]
Then using the definition of $h$, for all $i\in\NN$:
\[h_i(\preq^{\pi}_1,\ldots , \preq^{\pi}_n,\preq^{\pi}_{n+1})=f_i\vecopi=\pi(f_i\veco)=\pi(h_i(\preq_1,\ldots , \preq_n,\preq_{n+1}))\]
And for agent $n+1$:
\[h_{n+1}(\preq^{\pi}_1,\ldots , \preq^{\pi}_n,\preq^{\pi}_{n+1})=\MM\setminus \left(\bigcup_{i=1}^n f_i\vecopi \right)=\MM\setminus \left(\bigcup_{i=1}^n \pi(f_i\veco \right)\]
\[=\MM\setminus \pi\left(\bigcup_{i=1}^n f_i\vecopi \right)=\pi\left(\MM\setminus \left(\bigcup_{i=1}^n f_i\veco \right)\right)=\pi\left(h_{n+1}(\preq_1,\ldots , \preq_n,\preq_{n+1})\right)\]
In both cases we have shown that the requirement for neutrality holds for $h$ as well. In total we have shown that $h$ is a truthful, non-bossy, and neutral partition mechanism, therefore by \cref{strict-partition} there exists $\qpa\in \Q{n+1}{m}$ such that for all profiles $(\preq_1,\ldots , \preq_n,\preq_{n+1})\in \PP^{n+1}$:
\[h(\preq_1,\ldots , \preq_n,\preq_{n+1})=\SQqpa(\preq_1,\ldots , \preq_n,\preq_{n+1})\]
From the definition of $h$ we know that the valuation $\preq_{n+1}$ is not considered at all when choosing the allocation, meaning that agent $n+1$ cannot ever affect her bundle. Then we know for sure that $p'_{n+1}=n+1$, meaning that agent $n+1$ is the last agent in the picking order and gets the remaining goods after the other agents have chosen. This is because either agent $n+1$ has a quota larger than zero, and then she must be last otherwise she could affect her outcome by changing her demand, or she has a quota of zero, but then she is still last since we assumed in \cref{quota-ordering-pair} that agents with quotas of zero are ordered by index, and $n+1$ is the largest index.

We have shown that $p'_{n+1}=n+1$. Let $\qp\in\Qnm$ be defined s.t. $p_j=p'_j$ and $q_j=q'_j$ for all indices $j\in[n]$. Since in serial-quota mechanisms the last agent cannot affect what the other agents receive, for all profiles $(\preq_1,\ldots , \preq_n,\preq_{n+1})\in \PP^n$:
\[\SQqpa_{-(n+1)}(\preq_1,\ldots , \preq_n,\preq_{n+1})=\SQqp\veco\]
Then, for any profile $\veco\in \PP^n$, choosing some arbitrary $\preq_{n+1}\in \PP$:
\[f\veco=h_{-(n+1)}(\preq_1,\ldots , \preq_n,\preq_{n+1})=\SQqpa_{-(n+1)}(\preq_1,\ldots , \preq_n,\preq_{n+1})=\SQqp\veco\]
Proving that $f$ is a $\qp$-serial-quota mechanism.
\end{proof}
}

\section{Cardinal Allocation Mechanisms}
\label{appendix:cardinal-allocations}

\begin{definition} [{Cardinal allocation mechanism}]
A \emph{cardinal allocation mechanism} $f$ of a set $\MM$ of goods among a set $\NN$ of $n$ agents, each with a valuation from class $\VV$, is a (deterministic) function that maps a profile of $n$ valuations that are each in $\VV$, to an allocation $A\in \ALLOCsM$ of the set of goods $\MM$ among
the agents. That is, for $f:\VV^n\rightarrow \ALLOCsM$, the set that $f$ allocates to agent $i$ when the 
valuation profile is $\vecv\in \VV^n$
is $f_i(\vecv)$.
\end{definition}
When  it is clear from the context that $f$ is a \emph{cardinal} allocation mechanism (as it is defined over valuations)  
we sometimes simply say that it is an \emph{allocation mechanism} or simply a \emph{mechanism}.

We next present several standard definitions of properties of cardinal allocation mechanisms. For the following definitions, let $f:\VV^n\rightarrow \ALLOCsM$ be a cardinal allocation mechanism of a set  $\MM$ of $m$ goods  to a set $\NN$ of $n$ agents, each with valuation in the class $\VV$.

\begin{definition}
    [{Truthfulness in cardinal allocation mechanisms}]
    Cardinal allocation mechanism  $f$ is \emph{truthful} (in dominant strategies), if for every  agent $i\in\NN$, every 
    valuation 
    profile $\vecv= (v_i,v_{-i})\in \VV^{n}$, and every alternate {reported valuation}     $v'_i\in \VV$ of agent $i$, it holds that  \[ f_i(v'_i,v_{-i}) \preq_{v_i} f_i(v_i,v_{-i})\]
\end{definition}
That is, an agent cannot improve her utility by misreporting her valuation. {Note that this is a stronger requirement than just asking for misreports that change the induced preferences
not to be beneficial: misreporting cannot be beneficial even if {the valuation changed but} the induced preference remains the same.} 

\begin{definition} [{Non-bossiness in cardinal allocation mechanisms}]
{
Consider a class $\VV$ of strict valuations. Cardinal allocation mechanism $f$ is \emph{non-bossy} if for every agent $i\in\NN$, every valuation profile $\vecv= (v_i,v_{-i})\in \VV^{n}$, and every alternate valuation function $v_i\in \VV$, it holds that  \[ f_i(v'_i,v_{-i})= f_i(v_i,v_{-i}) \implies f(v'_i,v_{-i}) = f(v_i,v_{-i})\]
}
\end{definition}
A cardinal allocation mechanism is non-bossy if {a change in the {reported valuation} of an agent does not change her own allocation, it must not change any of the other agents' allocations.}

To define neutrality we need to define the valuation resulted from applying a permutation on the goods.
\begin{definition} [{Permuted valuation}]
    Given a permutation $\pi\in\perM$, for any valuation $v$, let $v^{\pi}$ denote the valuation defined $v^{\pi}(S)=v(\pi^{-1}(S))$ for every subset $S\subseteq\MM$.    
\end{definition}

The above requirement can be equivalently stated: $v^{\pi}(\pi(S))=v(S)$ for every subset $S\subseteq\MM$. Note that {it implies that} for every valuation $v$ and permutation $\pi\in\perM$ it holds that $\preq_{v^{\pi}}=\preq_v^{\pi}$\. That is, the preference induced by a valuation after permuting the goods is the same as the permutation of the induced preference.
We say that a set of valuations $\VV$ is \emph{permutations-closed} if it is closed under permutations, that is, for any permutation  $\pi\in \perM$ and $v\in \VV$ it holds that the preference $v^{\pi}$ is in $\VV$ as well.

\begin{definition} [{Neutrality in cardinal allocation mechanisms}]
        Consider any permutations-closed valuation class $\VV$ of strict valuations. 
    Cardinal allocation mechanism $f$ is \emph{neutral} if for every permutation $\pi\in\perM$ and valuation profile $\vecv=(v_1,...,v_n)\in \VV^n$, it holds that:  \[f_i(v^{\pi}_1,...,v^{\pi}_n)=\pi(f(v_1,...,v_n))\]
\end{definition}

{We define a cardinal serial-quota mechanism (for valuations) as the result of the corresponding serial-quota mechanism on the induced preferences as follows:}
\begin{definition}
    {
    Let $f$ be a {cardinal} allocation mechanism of a set $\MM$ of $m$ goods to a set $\NN$ of $n$ agents with valuations from a class $\VV$ of strict valuations. For a quota-ordering pair $\qp\in\Qnm$, we say $f$ is a \emph{$\qp$-serial-quota mechanism on $\VV$} if for every valuation vector $\vecv\in\VV^n$:
    \begin{equation}
\label{eq:cardinal-to-ordinal}
    f(v_1,...,v_n)=\SQqp\vecord{v}
\end{equation}
    
    }
\end{definition}

Meaning, the allocation of {a cardinal serial-quota mechanism} is decided entirely based on the induced ordinal preferences of the agents.

Let $\VV\subseteq \Vmon(\MM)$ be some valuation class, that may contain non-strict valuations. For a sub-class $\widehat{\VV}\subseteq \VV$ that contains only strict valuations, we say that {the} cardinal allocation mechanism $f$ is a \emph{serial-quota mechanism on $\widehat{\VV}$} if there exists a quota-ordering pair $\qp\in\Qnm$ s.t. $f$ satisfies \cref{eq:cardinal-to-ordinal} for every valuation profile $\vecv\in \widehat{\VV}^n$.

We say that {the} cardinal allocation mechanism $f$ is a \emph{serial-quota mechanism} if there exists a quota-ordering pair $\qp\in\Qnm$ s.t. $f$ is a $\qp$-serial-quota mechanism. Similarly, we say $f$ is a  \emph{serial-quota mechanism  on $\widehat{\VV}$} if there exists a quota-ordering pair $\qp\in\Qnm$ s.t. $f$ is a $\qp$-serial-quota mechanism on $\widehat{\VV}$.

\subsection{\texorpdfstring{Missing Proofs from \cref{sec:implications-val}}
                        {Missing Proofs from Section~\ref{sec:implications-val}}}
\label{appendix:subsec:implications-val}

\setOrderingEquivalence*

\begin{proof}
    We prove by induction on $i$ from $0$ to $n$ that $f(v_1,...,v_n)=f(u_1,...,u_i,v_{i+1},...,v_n)$. The base case $i=0$ is trivial. Let $i\in[n]$, assuming that  $f(v_1,...,v_n)=f(u_1,...,u_{i-1},v_{i},...,v_n)$, we prove the induction claim. Valuation profile $(u_1,...,u_i,v_{i+1},...,v_n)$ is obtained from $(u_1,...,u_{i-1},v_{i},...,v_n)$ by having agent $i$ swap her reported valuation from $v_i$ to $u_i$. Truthfulness in both directions implies: 
    \[f_i(u_1,...,u_i,v_{i+1},...,v_n) \preq_{v_i} f_i(u_1,...,u_{i-1},v_{i},...,v_n)\] \[f_i(u_1,...,u_{i-1},v_{i},...,v_n) \preq_{u_i} f_i(u_1,...,u_i,v_{i+1},...,v_n)\]
    Yet we assumed that $\preq_{v_i}=\preq_{u_i}$ and that the valuations are strict, so this is only possible if:
    \[f_i(u_1,...,u_i,v_{i+1},...,v_n)= f_i(u_1,...,u_{i-1},v_{i},...,v_n)\]
    Then non-bossiness implies $f(u_1,...,u_i,v_{i+1},...,v_n)= f(u_1,...,u_{i-1},v_{i},...,v_n)=f(v_1,...,v_n)$, proving the induction claim. Taking the end of the induction $i=n$ we have $f(\vecv)=f(\vecu)$.
\end{proof}

\cardinalstrict*

\begin{proof}
        It is easy to see that when valuations are strict, any serial-quota mechanism is truthful, non-bossy, and neutral. We will prove the converse. Assume that $f$ is truthful, non-bossy, and neutral.  \cref{setOrderingEquivalence} implies that the allocation $f$ outputs depends only on the induced ordinal preferences of the agents. Then we can define an ordinal allocation mechanism $h$ of the set of goods $\MM$ among $n$ agents $\NN$ with preferences from the class $\PP_{\VV}$ s.t. for every valuation vector $\vecv\in\VV$:
    \[f(v_1,...,v_n)=h(\preq_{v_1},...,\preq_{v_n})\]
   {Notice that since $\VV$ is class of strict valuations that is  permutations-closed and induces every lexicographic preference, the class of preferences $\PP_{\VV}$ is
   a class of strict preferences that is also permutations-closed and contains every lexicographic preference, meaning that $\PP_{\VV}\in \FFn(\MM)$.}
   We will prove that the truthfulness, non-bossiness, and neutrality of $f$ implies that the same holds for $h$. Truthfulness and non-bossiness of $h$ are trivial from the above definition. Neutrality is implied from the relationship $\preq_{v^{\pi}}=\preq_v^{\pi}$:
    \[h\left(\preq_{v_1}^{\pi},...,\preq_{v_n}^{\pi}\right)=h\vecord{v^{\pi}}=f\vecoper{v}{\pi}=\pi(f(v_1,...,v_n))=\pi(h\vecord{v})\] 
    Then \cref{theorem:strict} implies that $h$ is a $\qp$-serial-quota mechanism for some $\qp\in\Qnm$. Then, for every valuation vector $\vecv\in \VV^n$:
    \[f(v_1,...,v_n)=h(\preq_{v_1},...,\preq_{v_n})=\SQqp(\preq_{v_1},...,\preq_{v_n})\]
    
    Proving that $f$ is a $\qp$-serial-quota mechanism on $\VV$.
\end{proof}

\section{\texorpdfstring{Missing Proofs from \cref{sec:implications-fair}}
                        {Missing Proofs from Section~\ref{sec:implications-fair}}}
\label{appendix:implications-fair}

\quotaefone*
\begin{proof}
    Let $f$ be a serial-quota mechanism, then there exists a quota-ordering pair $\qp\in\Qnm$ such that $f$ is a $\qp$-serial-quota mechanism. We assume that $f$ is EF1 and show that either $q_i\leq 1$ for every index $i\in [n]$ or $\vecp{q}=(1,1,1,\ldots,1,1,2)$.
    
    Assume for contradiction that $q_i> 1$ for some index $i\in[n-1]$. Fix some strict additive valuations for agents $p_1,...,p_{n-1}$, note that the allocation of these agents does not depend on the preference of the last agent $p_n$. Consider the case where the last agent $p_n$ only values the goods received by agent $p_i$, and all other goods are valued at $0$. Then whatever allocation she gets, she will envy the allocation of agent $p_i$, even when any one of the $q_i>1$ goods is removed, contradicting the assumption that $f$ is EF1. We conclude that $q_i\leq 1$ for every index $i\in[n-1]$.

    Next we show that if $q_i=0$ for some $i\in[n-1]$ then $q_n\leq 1$. Assume for contradiction that $q_n>2$. Consider the case when all agents value all goods identically. Agent $p_i$ receives zero goods while agent $p_n$ receives at least two. Then agent $p_i$ will envy the allocation of agent $p_n$, even when one of the $q_n>1$ goods is removed, contradicting the assumption that $f$ is EF1.

    Finally, we show that if $q_i=1$ for every $i\in[n-1]$ then $q_n\leq 2$. Assume for contradiction that $q_n>2$. Again consider the case when all agents value all goods identically. Agent $p_i$ receives one good while agent $p_n$ receives at least three, so again agent $p_i$ will envy the allocation of agent $p_n$, even when one of the $q_n>1$ goods is removed, contradicting the assumption that $f$ is EF1. 
    \end{proof}

\end{document}